\documentclass[12pt]{article}
\usepackage{amssymb,amsmath,amsthm,fullpage,graphicx, bm, bbm, enumerate, setspace, mathtools, float, sectsty}
\usepackage[numbers]{natbib}
\usepackage[margin=1in]{geometry}
\usepackage{url}
\usepackage{color}
\usepackage{tikz}
\usepackage{booktabs}
\usepackage{caption}
\usepackage{subcaption}

\allowdisplaybreaks

\usepackage{hyperref}
\hypersetup{
    colorlinks=true,
    linkcolor=blue,
    filecolor=blue,      
    urlcolor=blue,
    citecolor=blue
}

\subsubsectionfont{\normalfont\itshape}

\definecolor{tan}{HTML}{ECCBAE}
\definecolor{dblue}{HTML}{046C9A}
\definecolor{lblue}{HTML}{ABDDDE}

\theoremstyle{plain}
\newtheorem{theorem}{Theorem}[section]
\newtheorem*{theorem*}{Theorem}

\newtheorem*{lemma*}{Lemma}

\newtheorem{proposition}{Proposition}[section]

\theoremstyle{definition}

\newtheorem{example}{Example}[section]

\theoremstyle{remark}

\def\T{{ \mathrm{\scriptscriptstyle T} }}

\begin{document}
\mathtoolsset{centercolon}

\title{Mixture representations and Bayesian nonparametric inference for likelihood ratio ordered distributions}

\date{} 

\author{Michael Jauch,  \textit{Florida State University} \\ Andrés F. Barrientos, \textit{Florida State University} \\ Víctor Peña, \textit{Universitat Politècnica de Catalunya} \\ David S. Matteson, \textit{Cornell University}}

\maketitle

\date{\vspace{-5ex}}
\maketitle

\begin{abstract}
In this article, we introduce mixture representations for likelihood ratio ordered distributions. Essentially, the ratio of two probability densities, or mass functions, is monotone if and only if one can be expressed as a mixture of one-sided truncations of the other. To illustrate the practical value of the mixture representations, we address the problem of density estimation for likelihood ratio ordered distributions. In particular, we propose a nonparametric Bayesian solution which takes advantage of the mixture representations. The prior distribution is constructed from Dirichlet process mixtures and has large support on the space of pairs of densities satisfying the monotone ratio constraint. Posterior consistency holds under reasonable conditions on the prior specification and the true unknown densities. To our knowledge, this is the first posterior consistency result in the literature on order constrained inference. With a simple modification to the prior distribution, we can test the equality of two distributions against the alternative of likelihood ratio ordering. We develop a Markov chain Monte Carlo algorithm for posterior inference and demonstrate the method in a biomedical application.
\end{abstract}



\clearpage

\section{Introduction} \label{intro}

Let $F$ and $G$ be distribution functions with densities, or mass functions, $f$ and $g.$ We say that $F$ is smaller than $G$ in the likelihood ratio order, denoted $F \leq_{\text{LR}} G,$ if $f/g$ is monotone non-increasing. The likelihood ratio order is a prominent example of a stochastic order \citep{Shaked2007}. The most familiar stochastic orders formalize, in different ways, the intuitive idea that one random variable, or distribution, is larger than another. Monotone likelihood ratios have important implications in statistical theory \citep{Karlin1956, Butucea2023} and feature in applications of statistics to economics and finance \citep{Roosen2004, Beare2015} as well as biology, medicine, and public health \citep{Dykstra1995, Carolan2005, Yu2017, Westling2023}. The likelihood ratio order is also pertinent to many problems in applied probability, stochastic processes, and operations research \citep{Shaked2007}.

This article introduces mixture representations for likelihood ratio ordered distributions. Essentially, we show that $f/g$ is non-increasing if and only if $f$ can be expressed as a mixture of one-sided truncations of $g.$ The mixing distribution, which assigns weights to the truncation points, is unique and tractable. A symmetric counterpart to this result provides a mixture representation for $g.$ The mixture representations shed light on the convex geometry of the likelihood ratio order and immediately lead to Khintchine-type results \citep{Khintchine1938, Williamson1956, Shepp1962}. 

To illustrate the practical value of the mixture representations, we address the problem of density estimation for likelihood ratio ordered distributions. Suppose we observe two independent samples $X_1, \ldots, X_{n} \sim F$ and $Y_1, \ldots, Y_{m} \sim G$ with $F \leq_{\text{LR}} G$. How can we flexibly estimate the densities $f$ and $g$ subject to the constraint that the ratio $f /g$ is monotone? 

As motivation for this problem, and following \citet{Yu2017} and \citet{Westling2023},
consider the relationship between a continuous explanatory variable $Z$ and a binary response variable $D.$ For example, $Z$ could be the value of a biomarker and $D$ could indicate the presence of a disease. In this situation, one commonly fits a generalized linear model, in which case $\text{pr}(D=1 \mid Z=z)$ is a monotonic function of $z.$ An alternative nonparametric analysis can be carried out via density estimation. Let $f$ and $g$ denote the densities of $Z$ conditional on $D=0$ and $D=1,$ respectively. By Bayes' theorem, 
\begin{align*} 
    \text{pr}(D=1 \mid Z=z) = \frac{g(z)\text{pr}(D=1)}{g(z)\text{pr}(D=1) + f(z)\text{pr}(D=0)} 
    = \left\{1 + \frac{f(z)}{g(z)}\frac{\text{pr}(D=0)}{\text{pr}(D=1)} \right\}^{-1}.
\end{align*}
To ensure the nonparametric analysis respects the monotonicity of {$\text{pr}(D=1 \mid Z=z),$} 
the ratio $f/g$ must be monotone. In many applications, the results may be difficult to interpret or act upon without such a constraint. Density estimation for likelihood ratio ordered distributions can also be applied in several seemingly unrelated statistical tasks, including estimating receiver operating characteristic curves \citep{Yu2017}, comparing trends in nonhomogenous Poisson processes \citep{Dykstra1995}, and accounting for selection bias.  

Despite the importance of the density estimation problem, there has been little work addressing it. Prior to the present article, only \citet{Yu2017} had proposed a solution, combining a reparametrization with the maximum smoothed likelihood technique of \citet{Eggermont1995a}. Recently, \citet{Hu2023} presented a method based on Bernstein polynomials. Both works established the consistency of their estimators, but left open the question of uncertainty quantification. There is a substantial and growing frequentist literature on related estimation \citep{Dykstra1995, Westling2023, Mosching2022} and hypothesis testing \citep{Roosen2004, Beare2015, Beare2019} problems. There are, as far as we are aware, no comparable Bayesian methods for likelihood ratio ordered distributions. The Bayesian literature addressing stochastic ordering constraints is focused primarily on the usual stochastic order \citep{Gelfand2001, Hoff2003a, Hoff2003b, Karabatsos2007, Dunson2008, Wang2011a}. 

We propose a nonparametric Bayesian solution to the density estimation problem which takes advantage of the mixture representations. Our strategy is to assign a prior to $g$ and the mixing density over truncation points. By doing so, we induce a prior distribution for $f$ and $g$ such that $f/g$ is monotone non-increasing almost surely. The prior distribution is constructed from Dirichlet process mixtures \citep{Lo1984} and has large support on the space of pairs of densities satisfying the monotone ratio constraint.  Posterior consistency holds under reasonable conditions on the prior specification and the true unknown densities. To our knowledge, this is the first posterior consistency result in the literature on order constrained inference. We develop a Markov chain Monte Carlo algorithm for posterior inference, enabling us to calculate, up to Monte Carlo error, posterior means and credible intervals for all parameters of interest. With a simple modification to the prior distribution, we can also calculate the posterior probability of the null hypothesis $H_0: F = G$ versus the alternative $H_1: F \leq_{\text{LR}}G,$ addressing a testing problem of interest in economics \citep{Roosen2004, Beare2015, Beare2019} and statistics \citep{Dykstra1995, Carolan2005}. Following \citet{Westling2023}, we apply our method to the evaluation of C-reactive protein as a means of diagnosing bacterial infection in children with systemic inflammatory response syndrome.  Proofs appear in the supplementary material, and \texttt{R} code can be found at \href{https://github.com/michaeljauch/lrmix}{https://github.com/michaeljauch/lrmix}.

\section{The mixture representations} \label{sec:mixtrep}

\subsection{Overview}

In this section, we present the mixture representations for likelihood ratio ordered distributions. We first consider finite discrete distributions and then move on to absolutely continuous distributions. Count distributions with countably infinite support are considered in the supplementary material. For each setting, there are two symmetric results. One expresses the density, or mass function, of the smaller distribution as a mixture, while the other expresses the density, or mass function, of the larger distribution as a mixture. Throughout the section, we present concrete examples to illustrate the results, build intuition, and make connections to previous work. 

\subsection{Finite discrete distributions}

Let $\mathcal{I} = \{x_1, \ldots, x_d\} \subset \mathbbm{R}$ with $x_1 < \ldots < x_d.$ Suppose $f$ and $g$ are probability mass functions with $\sum_{x_i \in \mathcal{I}} f(x_i) = 1, \sum_{x_i \in \mathcal{I}} g(x_i) = 1,$ and $g > 0$ on $\mathcal{I}.$ Let $F$ and $G$ be the corresponding distribution functions. For each $x_j \in \mathcal{I} \setminus \{x_d\},$ let $g^{x_j}$ denote the mass function that results from truncating $g$ above $x_j.$ More precisely, $g^{x_j}(x_i) = g(x_i)\mathbbm{1}_{(-\infty, x_j]}(x_i)/G(x_j)$ for $x_i \in \mathcal{I}.$ The superscript alludes to truncation above. We reserve the subscript for truncation below. Let $G^{x_j}$ be the corresponding distribution function.

With this notation in place, we can state the first result on mixture representations for likelihood ratio ordered distributions. 
\begin{theorem}[$f$ as a mixture] \label{thm:fmixture_finite}   The ratio $f/g$ is non-increasing on $\mathcal{I}$ if and only if there exists $\theta \in [0,1]$ and a mass function $u$ with $\sum_{x_j \in \mathcal{I} \setminus \{x_d\}} u(x_j) = 1$ 
such that 
\begin{align} \label{fmixture_finite}
    f(x_i) &= \theta g(x_i) + (1-\theta) \sum_{x_j \in \mathcal{I} \setminus \{ x_d \}}  g^{x_j}(x_i) u(x_j), \quad x_i \in \mathcal{I}.
\end{align} Suppose a mixture representation \eqref{fmixture_finite} exists. Then $\theta = f(x_d)/g(x_d).$ If $\theta \in [0,1),$ $u$ is also uniquely determined with 
\begin{align*}
    u(x_j) = \frac{G(x_j)}{1-\theta}\left\{\frac{f(x_j)}{g(x_j)} - \frac{f(x_{j+1})}{g(x_{j+1})}\right\}, \quad x_j \in \mathcal{I}\setminus \{x_d\}.
\end{align*}
\end{theorem}

Theorem \ref{thm:fmixture_finite} provides a characterization of the likelihood ratio order in terms of random truncations. The likelihood ratio ordering $F \leq_{\text{LR}} G$ is equivalent to the existence of $\theta \in [0,1]$ and a distribution function $U,$ corresponding to the mass function $u$ in \eqref{fmixture_finite},  such that the following procedure produces a random variable $X \sim F:$ With probability $\theta,$ draw $X \sim G$; otherwise, draw $S \sim U$ and then $X \mid S \sim G^S.$ This characterization suggests a promising approach to modeling likelihood ratio ordered distributions. Instead of modeling the constrained mass functions $(f,g)$ directly, we can model $(g, u, \theta)$ without constraints.  

Theorem \ref{thm:fmixture_finite} also highlights the simple convex geometry of the likelihood ratio order. Consider the collection of mass functions $\tilde{f}$ with $\sum_{x_i \in \mathcal{I}} \tilde{f}(x_i) = 1$ such that $\tilde{f}/g$ is non-increasing on $\mathcal{I}.$ The theorem tells us this collection is a $(d-1)$-dimensional simplex with $\{g, g^{x_{d-1}}, \ldots , g^{x_1}\}$ as its vertex set. From this perspective, Theorem \ref{thm:fmixture_finite} is reminiscent of Corollary 1 in \citet{Hoff2003a}, which concerns mixture representations for convex sets of probability measures. We can visualize the convex geometry of the likelihood ratio order by considering a low-dimensional example.

\begin{figure}
     \centering
     \begin{subfigure}[b]{0.45\textwidth}
         \centering
         \includegraphics[width=\textwidth]{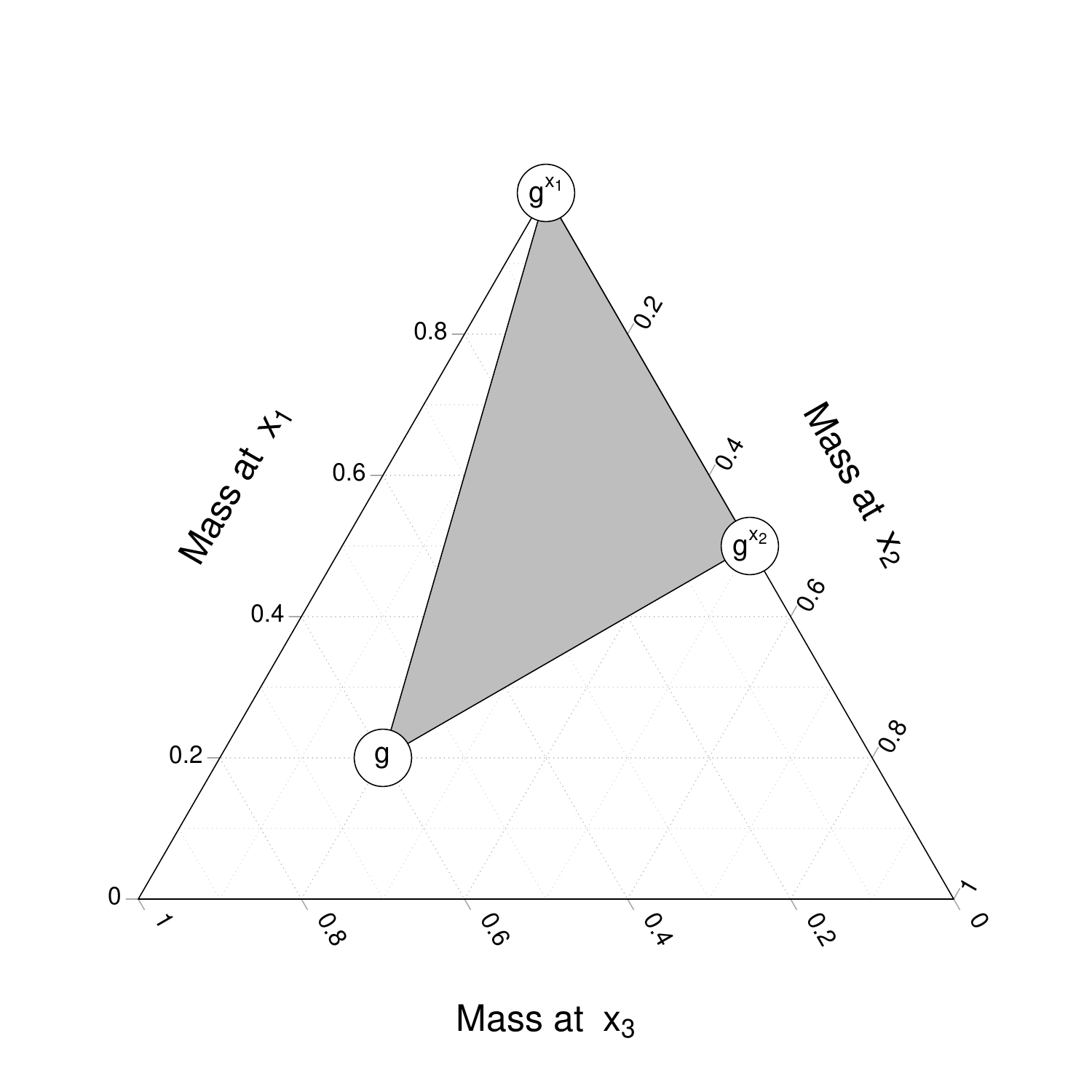}
         \caption{}
         \label{fig:ternary}
     \end{subfigure}
     \hfill
     \begin{subfigure}[b]{0.45\textwidth}
         \centering
         \includegraphics[width=\textwidth]{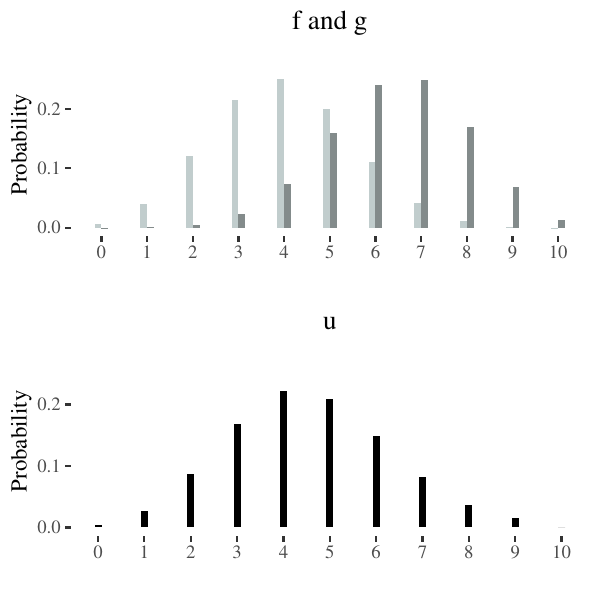}
         \caption{}
         \label{fig:binomial}
     \end{subfigure}
        \label{fig:ternary_plus_binomial}
        \caption{(a) The ternary plot described in Example \ref{ex:ternary}. (b) Top: the $\text{Binomial}(10,1/3)$ mass function $f$ (light grey) and the $\text{Binomial}(10,2/3)$ mass function $g$ (dark grey), as discussed in Example \ref{ex:binomial}. Bottom: the corresponding mass function $u.$}
\end{figure}

\begin{example} \label{ex:ternary} 
Suppose that $d=3$ with $g(x_1)= g(x_2)= 0.2$ and $ g(x_3)= 0.6.$ Mass functions supported on subsets of $\mathcal{I}$ can be identified with points in the two-dimensional probability simplex. For example, $g$ can be identified with the probability vector $[g(x_1), g(x_2), g(x_3)]^\T.$ Elements and regions of the probability simplex can in turn be visualized with a ternary plot. 
Figure \ref{fig:ternary} plots the elements corresponding to the mass functions $g, g^{x_{2}},$ and $g^{x_1}$ with white circles. The shaded region corresponds to the collection of mass functions $\tilde{f}$  with $\sum_{x_i \in \mathcal{I}} \tilde{f}(x_i) = 1$ such that $\tilde{f}/g$ is non-increasing on $\mathcal{I}.$ As we expect based on Theorem \ref{thm:fmixture_finite}, this region is triangular and its vertices correspond to the mass functions $g, g^{x_{2}},$ and $g^{x_1}.$
\end{example}

\begin{example} \label{ex:binomial} Let $F$ and $G$ correspond to $\text{Binomial}(10, 1/3)$ and $\text{Binomial}(10, 2/3)$ distributions, respectively. Then $F \leq_{\text{LR}} G.$ The top panel of Fig. \ref{fig:binomial} plots the two binomial mass functions, while the bottom panel plots the mass function $u.$ In this example, $\theta \approx  9.8 \times 10^{-4}.$
\end{example}

There is a symmetric counterpart to Theorem \ref{thm:fmixture_finite}. We now suppose $f$ rather than $g$ is positive on $\mathcal{I}.$ For each $x_j \in \mathcal{I} \setminus \{x_1\},$ let $f_{x_j}$ denote the mass function that results from truncating $f$ below $x_j.$ More precisely, $f_{x_j}(x_i) = f(x_i)\mathbbm{1}_{[x_j, \infty)}(x_i)/\{1 - F(x_{j-1})\}$ for $x_i \in \mathcal{I}.$ 

\begin{theorem}[$g$ as a mixture] \label{thm:gmixture_finite}
The ratio $g/f$ is non-decreasing on $\mathcal{I}$ if and only if there exists $\omega \in [0,1]$ and a mass function $v$ with $\sum_{x_j \in \mathcal{I}\setminus \{x_1\}} v(x_j) = 1$ such that
\begin{align} \label{gmixture_finite}
    g(x_i) &= \omega f(x_i) + (1-\omega) \sum_{x_j \in \mathcal{I}\setminus \{x_1\}} f_{x_j}(x_i) \, v(x_j), \quad x_{i} \in \mathcal{I}.
\end{align} Suppose a mixture representation \eqref{gmixture_finite} exists. Then $\omega=g(x_1)/f(x_1).$ If $\omega \in [0,1),$ $v$ is also uniquely determined with 
\begin{align*}
    v(x_j) = \frac{1-F(x_{j-1})}{1-\omega}\left\{\frac{g(x_j)}{f(x_j)} - \frac{g(x_{j-1})}{f(x_{j-1})} \right\}, \quad x_j \in \mathcal{I}\setminus \{x_1\}.
\end{align*}
\end{theorem}

\subsection{Absolutely continuous distributions}

A natural question is whether analogous mixture representations exist for absolutely continuous distributions that satisfy a likelihood ratio order. The answer is yes, provided that the density ratio is sufficiently regular.

Let $I=(a,b)$ with $a,b \in \mathbbm{R} \cup \{-\infty, \infty\}.$ Suppose $f$ and $g$ are density functions with $\int_{I} f(x)\, dx = 1,$ $\int_{I} g(x)\, dx = 1,$ and $g > 0$ on $I.$ Let $F$ and $G$ be the corresponding distribution functions. For each $s \in I,$ let $g^s$ denote the truncated density function with $g^s(x) = g(x) \mathbbm{1}_{(-\infty,s]}(x)/G(s)$ for $x \in I.$ Following \citet{Leoni2017}, we say that a function $h: I \to \mathbbm{R}$ is locally absolutely continuous on $I$ if it is absolutely continuous on every compact interval $J \subset I,$ and we denote the space of all such functions as $AC_\text{loc}(I).$ 

With this notation in place, we can state an analogue of Theorem \ref{thm:fmixture_finite} for absolutely continuous distributions. 
\begin{theorem}[$f$ as a mixture] \label{thm:fmixture_ac}
The ratio $f/g$ is non-increasing and locally absolutely continuous on $I$ if and only if there exists $\theta \in [0,1]$ and an absolutely continuous distribution function $U$ with $U(a+)=0$ and $U(b-)=1$ such that
\begin{align} \label{fx_mixture}
    f(x) &= \theta g(x) + (1-\theta) \int_a^b g^s(x) \, U'(s) \,  ds, \quad x\in I.
\end{align} Suppose a mixture representation \eqref{fx_mixture} exists. Then $\theta=\lim_{x \uparrow b} f(x)/g(x).$ If $\theta \in [0,1),$ $U$ is also uniquely determined with 
\begin{align*}
    U(x) = \frac{G(x)}{1-\theta}\left\{\frac{F(x)}{G(x)} - \frac{f(x)}{g(x)}\right\}, \quad x \in I.
\end{align*}
\end{theorem}

\begin{figure}
     \centering
     \begin{subfigure}[b]{0.45\textwidth}
         \centering
         \includegraphics[width=\textwidth]{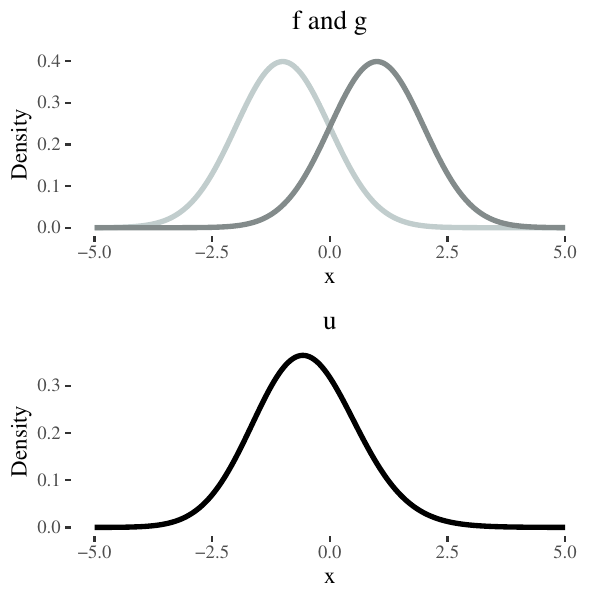}
         \caption{}
         \label{fig:normal}
     \end{subfigure}
     \hfill
     \begin{subfigure}[b]{0.45\textwidth}
         \centering
         \includegraphics[width=\textwidth]{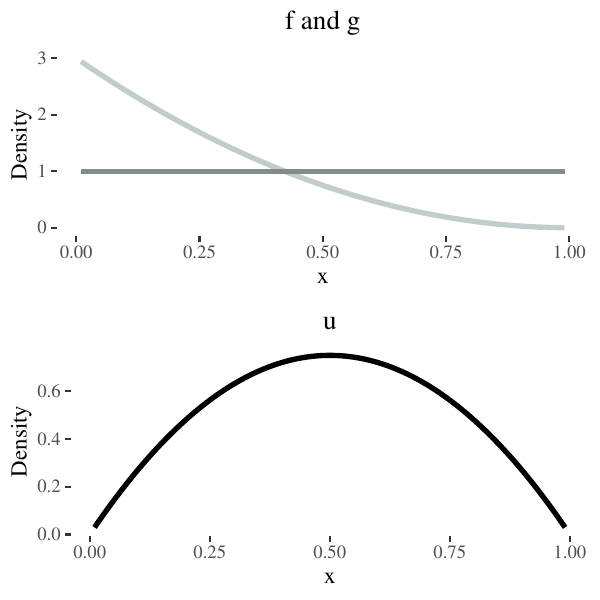}
         \caption{}
         \label{fig:beta}
     \end{subfigure}
    \caption{(a) Top: the $\text{Normal}(-1,1)$ density $f$ (light grey) and the $\text{Normal}(1,1)$ density $g$ (dark grey), as discussed in Example \ref{ex:normal}. Bottom: the corresponding density $u.$ (b) Top: the $\text{Beta}(1,3)$ density $f$ (light grey) and the $\text{Beta}(1,1)$ density $g$ (dark grey), as discussed in Example \ref{ex:beta}. Bottom: the corresponding density $u.$}
\end{figure}

We offer a couple of initial remarks related to Theorem \ref{thm:fmixture_ac}. First, any density $u$ of $U$ satisfies $u(x) = G(x)\left\{-f(x)/g(x)\right\}'/(1-\theta)$ for almost all $x \in I.$ The second remark concerns the regularity condition. Local absolute continuity of the ratio $f/g$ is critical to the existence of the mixture representation \eqref{fx_mixture}, but it is not obvious how to verify that it holds for a particular example. The following proposition gives a simple sufficient condition.
\begin{proposition} \label{prop:suff_cond_fx}
Suppose $f/g$ is continuous and non-increasing on $I.$ If $f/g$ is differentiable on $I$ except at countably many points, then $f/g \in AC_\text{loc}(I).$
\end{proposition}

\begin{example} \label{ex:normal}
Let $F$ and $G$ correspond to $\text{Normal}(-1, 1)$ and $\text{Normal}(1, 1)$ distributions, respectively. Then $F \leq_{\text{LR}} G$ and $f/g \in AC_\text{loc}(I).$ The top panel of Fig. \ref{fig:normal} shows a plot of the two normal densities, while the bottom panel shows a plot of the density $u.$ In this example, $\theta =0.$ 
\end{example}

Theorem \ref{thm:fmixture_ac} is akin to previous work on mixture representations for unimodal and monotone densities. 
The results of \citet{Shepp1962} and \citet{Khintchine1938} imply that unimodal densities can be characterized as mixtures of uniform densities. See \citet{Feller1971}, \citet[p.~171]{Devroye1986}, or \citet{Jones2002} for discussion. \citet{Williamson1956} provided a similar characterization of $k$-monotone densities. As a simple illustration of the connection between these results and our own, suppose that $g$ is a uniform density on a bounded interval $I.$ In that case, the ratio $f/g \propto f$ on $I.$ It follows from Theorem \ref{thm:fmixture_ac} that $f$ is non-increasing and locally absolutely continuous on $I$ if and only if $f$ is a mixture of uniform densities. The following example makes this more concrete.

\begin{example} \label{ex:beta}
Let $F$ and $G$ correspond to $\text{Beta}(1, 3)$ and $\text{Beta}(1, 1)$ distributions, respectively. Then $F \leq_{\text{LR}} G$ and $f/g \in AC_\text{loc}(I).$ The top panel of Fig. \ref{fig:beta} shows a plot of the two beta densities, while the bottom panel shows a plot of the density $u.$ In this example, $\theta = 0.$ 
\end{example}

Again, there is a symmetric counterpart to Theorem \ref{thm:fmixture_ac}. We now suppose $f$ rather than $g$ is positive on $I.$ For each $s \in I,$ let $f_s$ denote the truncated density function with $f_s(x) = f(x)\mathbbm{1}_{[s,\infty)}(x)/\{1 - F(s)\}$ for $x\in \mathbbm{R}.$

\begin{theorem}[$g$ as a mixture] \label{thm:gmixture_ac}
The ratio $g/f$ is non-decreasing and locally absolutely continuous on $I$ if and only if there exists $\omega \in [0,1]$ and an absolutely continuous distribution function $V$ with $V(a+) = 0$ and $V(b-)=1$ such that
\begin{align} \label{gx_mixture}
    g(x) &= \omega f(x) + (1-\omega) \int_a^b f_s(x) \, V'(s) \,  ds, \quad x\in I.
\end{align} Suppose a mixture representation \eqref{gx_mixture} exists. Then $\omega = \lim_{x\downarrow a} g(x)/f(x).$ If $\omega \in [0,1),$ $V$ is also uniquely determined with 
\begin{align*}
    1 - V(x) = \frac{1-F(x)}{1-\omega}\left\{\frac{1-G(x)}{1-F(x)} - \frac{g(x)}{f(x)}\right\}, \quad x \in I.
\end{align*}
\end{theorem}

Any density $v$ of $V$ satisfies $v(x) = \left\{1-F(x)\right\}\left\{g(x)/f(x)\right\}'/(1-\omega)$ for almost all $x \in I.$ The following proposition, analogous to Proposition \ref{prop:suff_cond_fx}, provides a simple sufficient condition for verifying that $g/f$ is locally absolutely continuous on $I.$ 
\begin{proposition} \label{prop:suff_cond_gx}
Suppose $g/f$ is continuous and non-decreasing on $I.$ If $g/f$ is differentiable on $I$ except at countably many points, then $g/f \in AC_\text{loc}(I).$
\end{proposition}

\section{The density estimation problem} \label{dens_est_section}

\subsection{Overview}

In this section, we revisit the density estimation problem in light of Theorem \ref{thm:fmixture_ac}. We propose a prior distribution based on the mixture representation \eqref{fx_mixture} and establish that it meets the requirements laid out in \citet{Ferguson1973} for prior distributions in nonparametric problems: It has large support on the space of pairs of densities satisfying the monotone ratio constraint, and it leads to computationally tractable posterior inference. Furthermore, posterior consistency holds under reasonable conditions on the prior specification and the true unknown densities. To our knowledge, this is the first posterior consistency result in the literature on order constrained inference. We evaluate the finite sample performance of the proposed density estimation method through a simulation study. Then we apply our method to the evaluation of C-reactive protein as a means of diagnosing bacterial infection in children with systemic inflammatory response syndrome.

\subsection{A prior from Dirichlet process mixtures} \label{prior_description}

As in Section \ref{intro}, suppose we observe two independent samples $X_1, \ldots, X_{n} \sim F$ and $Y_1, \ldots, Y_{m} \sim G$ with $F \leq_{\text{LR}} G.$ Our goal is to estimate the densities $f$ and $g$ subject to the constraint that the ratio $f /g$ is non-increasing. The primary challenge, from a Bayesian perspective, is to define a sufficiently flexible prior distribution for $(f,g)$ such that $f/g$ is non-increasing almost surely. 

Our strategy takes advantage of the mixture representation from Theorem \ref{thm:fmixture_ac}. To do so, we must make mild assumptions on the densities $(f,g).$ Define $\text{LR}_{\text{mix}}(\mathbbm{R})$ as the set of all pairs of densities $(\tilde{f}, \tilde{g})$ on $\mathbbm{R}$ with $\tilde{g} > 0$ such that $\tilde{f}/\tilde{g}$ is non-increasing and locally absolutely continuous on $\mathbbm{R}.$ We assume $(f,g) \in \text{LR}_{\text{mix}}(\mathbbm{R}).$  Theorem \ref{thm:fmixture_ac} then guarantees there exists $\theta \in [0,1]$ and an absolutely continuous distribution function $U$ such that \eqref{fx_mixture} holds. If we assign a prior measure $\Pi$ to $(g, u, \theta),$ the map $T_1: (g, u, \theta) \mapsto (f,g)$ induces a prior measure $\Pi \circ T_1^{-1}$ on $\text{LR}_{\text{mix}}(\mathbbm{R}).$ 

We model the densities $g$ and $u$ flexibly using Dirichlet process mixtures. A nonparametric approach is especially natural in the case of $u$ which, as the density of the latent truncation points, is somewhat unintuitive. 
Let $\text{DP}(P_0, \alpha)$ denote the Dirichlet process with base measure $P_0$ and precision parameter $\alpha > 0,$ and let $\phi_{\mu, \sigma^2}$ denote the density of a Gaussian distribution with mean $\mu$ and variance $\sigma^2.$  We suppose $g$ and $u$ are location-scale mixtures of Gaussians and assign the mixing measures Dirichlet process priors: 
\begin{align*}
g(\cdot \mid P_1) = \int_{\mathbbm{R} \times \mathbbm{R}^+} \phi_{\mu,\sigma^2}(\cdot) P_1(d\mu,d\sigma^2), \quad P_1 \sim \text{DP}(P_{1,0}, \alpha_1) \\ 
u(\cdot \mid P_2) = \int_{\mathbbm{R} \times \mathbbm{R}^+} \phi_{\mu,\sigma^2}(\cdot) P_2(d\mu,d\sigma^2), \quad P_2 \sim \text{DP}(P_{2,0}, \alpha_2).
\end{align*}
Computational details are summarized in Section \ref{computation} and presented more fully in the supplementary material. 

The parameter $\theta$ is also unknown and therefore requires a prior distribution. Choosing a spike-and-slab prior distribution that assigns positive prior probability to the event $\theta = 1$ enables us to evaluate the posterior probability of the null hypothesis $H_0: F = G,$ equivalent to $\theta = 1,$ versus the alternative $H_1: F \leq_{\text{LR}}G,$ equivalent to $\theta \in [0,1).$ The mixture representation \eqref{fx_mixture} thus offers a convenient way of addressing the testing problem alluded to in Section \ref{intro}. We illustrate this approach to testing in an additional data example in the supplementary material.


\subsection{Large support}

Large support of the prior distribution is a defining characteristic of Bayesian nonparametric procedures and a necessary condition for the posterior distribution to concentrate around the truth. Under mild conditions, the prior measure $\Pi \circ T_1^{-1}$ constructed from Dirichlet process mixtures has full support on $\text{LR}_{\text{mix}}(\mathbbm{R}).$ Full support, in this context, means that $\Pi \circ T_1^{-1}$ assigns positive mass to all neighborhoods of any element of $\text{LR}_{\text{mix}}(\mathbbm{R}).$ 

The full support property is defined with respect to a topology on $\text{LR}_{\text{mix}}(\mathbbm{R}).$ A natural choice in density estimation problems is the topology induced by the $L_1$ metric \citep{Devroye1985}. Since $\text{LR}_{\text{mix}}(\mathbbm{R})$ consists of pairs of densities, we consider the topology induced by the $L_1$-based metric $d_1\{(h, \ell), (\tilde{h}, \tilde{\ell})\} = \|h - \tilde{h}\|_1 + \|\ell - \tilde{\ell}\|_1,$ where $h, \ell, \tilde{h},$ and $\tilde{\ell}$ are densities and $\|h - \tilde{h}\|_1 = \int_{\mathbbm{R}} |h(x) - \tilde{h}(x)|\, dx.$ The topology induced by the metric $d_1$ is a particular case of the integrated $L_1$-topology employed by \citet{Pati2013}.

The following theorem states the conditions under which the prior measure $\Pi \circ T_1^{-1}$ has full $d_1$-support on $\text{LR}_{\text{mix}}(\mathbbm{R}).$ 
\begin{theorem}\label{thm:support}
The prior measure $\Pi \circ T_1^{-1}$ has full $d_1$-support on $\text{LR}_{\text{mix}}(\mathbbm{R})$ provided that \\
\indent A1. $P_{1},$ $P_{2},$ and $\theta$ are independent; \\ 
\indent A2. $P_{1,0}$ and $P_{2,0}$ have full support on $\mathbbm{R}\times \mathbbm{R}^+;$ \\ 
\indent A3. The prior distribution for $\theta$ has full support on $[0,1].$
\end{theorem} 

The proof of Theorem \ref{thm:support}, which appears in the supplementary material, depends upon two important and nontrivial lemmas. The first concerns the continuity of the map $T_1: (g, u, \theta) \mapsto (f,g),$ while the second concerns the $L_1$-support of Dirichlet process location-scale mixtures of Gaussian distributions. The second lemma may be of independent interest. 

\subsection{Posterior consistency} \label{consistency}

A strong case can be made, both from a frequentist and a Bayesian perspective, for considering large sample properties of the posterior distribution and, in particular, posterior consistency \citep[p.~126]{Ghosal2017}. Posterior consistency holds when the posterior distribution concentrates around the true parameter of the data-generating distribution as the sample size grows. It guarantees the existence of a point estimator that is consistent in the usual frequentist sense and ensures that the data eventually swamp the prior distribution. Posterior consistency also implies a degree of robustness with respect to the choice of prior distribution and a merging of opinions with increasing information. Despite the strong case for considering large sample properties of the posterior distribution, we are not aware of any posterior consistency results in the literature on order constrained inference. 

We leverage Schwartz's theorem \citep{Schwartz1965, Ghosal2017}, a foundational result in Bayesian nonparametric statistics, to establish the posterior consistency of our density estimation procedure. Schwartz's theorem requires that the data are i.i.d. from a single density, while our density estimation problem involves independent samples from two densities. Thus, in order to apply Schwartz's theorem, we reformulate the density estimation problem using a strategy similar to that employed by \citet{Pati2013} and \citet{Barrientos2017} in the context of densities indexed by predictors. Let $\left(\mathfrak{f}, \mathfrak{g}\right) \in \text{LR}_{\text{mix}}(\mathbbm{R})$ 
and suppose we observe data $\left(Z_1, D_1\right), \ldots, \left(Z_\ell, D_\ell \right)$ where $\ell = n + m$  and each $Z_i$ is drawn from either $\mathfrak{f}$ or $\mathfrak{g}$ depending upon the value of the binary random variable $D_i.$ More precisely, suppose that
\begin{align*}
&Z_i \, \vert \, D_i = 0 \stackrel{\text{ind.}}{\sim} \mathfrak{f} \\  
&Z_i \, \vert \, D_i = 1 \stackrel{\text{ind.}}{\sim} \mathfrak{g} \\  
&D_i \stackrel{\text{i.i.d.}}{\sim} \text{Bernoulli}\left(q\right) 
\end{align*}
with $q$ known. After this reformulation, the data $\left(Z_1, D_1\right), \ldots, \left(Z_\ell, D_\ell \right)$ are an i.i.d. sample from the joint density 
\begin{align*}
\mathfrak{m}\left(z, d\right) = \left[q \cdot \mathfrak{g}(z) \right]^{d} \left[\left(1-q\right) \cdot \mathfrak{f}(z)\right]^{1-d}
\end{align*} and we are in a position to apply Schwartz's theorem. 

Before we do so, we must introduce notation and define relevant concepts. 
Let $$\text{LR}_\text{joint}(\mathbbm{R}) = \left\{ \tilde{m}\left(z,d\right) = \left[q \cdot \tilde{g}(z) \right]^{d} [\left(1-q\right) \cdot \tilde{f}(z)]^{1-d} \mathrel{\Big|}  (\tilde{f}, \tilde{g}) \in \text{LR}_\text{mix}\left(\mathbbm{R}\right) \right\},$$ i.e. $\text{LR}_\text{joint}(\mathbbm{R})$ is the set of joint densities $\tilde{m}$ constructed from a pair of likelihood ratio ordered densities $\left(\tilde{f}, \tilde{g}\right)\in \text{LR}_\text{mix}(\mathbbm{R})$ as described above. The posterior distribution associated with the random joint density $$m(z,d) = \left[q \cdot g(z) \right]^{d} \left[\left(1-q\right) \cdot f(z)\right]^{1-d}$$ is \textit{strongly consistent} at $\mathfrak{m}$ if, almost surely, the posterior probability of any neighborhood $N$ of $\mathfrak{m}$ converges to one as the sample size $\ell$ goes to infinity. (Following \citet[p.~123]{Ghosal2017}, ``strong" refers to almost sure convergence rather than the topology on the parameter space.)

The concept of posterior consistency is defined with respect to a specific topology. In this work, we focus on consistency with respect to the weak topology \citep[p.~507]{Ghosal2017} on $\text{LR}_\text{joint}(\mathbbm{R}).$ The map $T_2: \text{LR}_\text{mix}(\mathbbm{R}) \to \text{LR}_\text{joint}(\mathbbm{R})$ implied by the discussion above induces a topology on $\text{LR}_\text{mix}(\mathbbm{R}).$ For any neighborhood $N$ of $\mathfrak{m},$ we can define a neighborhood 
$T_2^{-1}\left(N\right)$ of $(\mathfrak{f}, \mathfrak{g}).$ Strong consistency of the posterior associated with $m$ at $\mathfrak{m}$ with respect to the weak topology on $\text{LR}_\text{joint}(\mathbbm{R})$ implies strong consistency of the posterior associated with $(f,g)$ at $(\mathfrak{f}, \mathfrak{g})$ with respect to the topology induced on $\text{LR}_\text{mix}(\mathbbm{R}).$ 

Schwartz's theorem provides a sufficient condition for strong consistency of the posterior associated with $m$ at $\mathfrak{m}$ with respect to the weak topology on $\text{LR}_\text{joint}(\mathbbm{R}):$ The prior distribution $\Pi \circ T_2^{-1}$ on $\text{LR}_\text{joint}(\mathbbm{R})$ must assign positive probability to any Kullback-Leibler neighborhood of $\mathfrak{m}.$ That is, for any $\epsilon > 0,$ we must have that 
\begin{align} \label{KLproperty}
\Pi \circ T_2^{-1}\left\{\text{KL}\left(\mathfrak{m},m\right) < \epsilon \right\} > 0
\end{align} where $\text{KL}(\cdot, \cdot)$ denotes the Kullback-Leibler divergence. 

The following theorem establishes conditions on $\Pi$ and $(\mathfrak{f}, \mathfrak{g})$ such that the Kullback-Leibler property \eqref{KLproperty} is satisfied and our density estimation procedure achieves posterior consistency. Provided that $\theta_0 = \lim_{x \uparrow \infty} \mathfrak{f}(x)/\mathfrak{g}(x) \in [0,1),$ there exists a density $\mathfrak{u}$ that uniquely determines $\mathfrak{f}$ through the mapping $T_1: (\mathfrak{g}, \mathfrak{u}, \theta_0) \mapsto (\mathfrak{f},\mathfrak{g}).$ Thus, conditions on $(\mathfrak{f},\mathfrak{g})$ may also be stated in terms of $(\mathfrak{g}, \mathfrak{u}, \theta_0).$ The proof, which appears in the appendix, is quite involved. 

\begin{theorem}\label{thm:consistency}
Let $\Pi$ be chosen to satisfy conditions A1, A2, and A3 of Theorem \ref{thm:support}. Suppose that \\ 
\indent B1. The densities $\left(\mathfrak{f}, \mathfrak{g}\right) \in \text{LR}_\text{mix}(\mathbbm{R})$ with $\theta_0 = \lim_{x \uparrow \infty} \mathfrak{f}(x)/\mathfrak{g}(x) \in [0,1);$ \\  
\indent B2. The density ratio $\mathfrak{f}/\mathfrak{g}$ is bounded above. 

\noindent Furthermore, suppose that $\mathfrak{g} , \mathfrak{u} $ satisfy the following conditions (adapted from Theorem 3.2 in \citet{tokdar2006posterior}): \\ 
\indent B3.\, The densities $\mathfrak{g} , \mathfrak{u} $ are continuous, nowhere zero, and bounded above; \\ 
\indent B4.\,  $\left|\int_{\mathbbm{R}} \mathfrak{g} (x) \log \mathfrak{g} (x) \, dx \right| <\infty$ and $\left|\int_{\mathbbm{R}} \mathfrak{u} (x) \log \mathfrak{u} (x) \, dx \right| <\infty$;\\
\indent B5.\,  There exist $\eta_1,\,\eta_2>0$ such that $\int_{\mathbbm{R}} |x|^{2(\eta_1+1)} \mathfrak{g} (x) \, dx <\infty$ and $\int_{\mathbbm{R}} |x|^{2(\eta_2+1)} \mathfrak{u} (x) \, dx <\infty$; \\
\indent B6.\, $\int_{\mathbbm{R}} \mathfrak{g} (x) \log \frac{\mathfrak{g} (x)}{\psi_1(x)} \, dx < \infty$ and $\int_{\mathbbm{R}} \mathfrak{u} (x) \log \frac{\mathfrak{u} (x)}{\psi_2(x)}  \, dx <\infty$ where $\psi_1(x) = \inf_{t\in[x-1,x+1]}\mathfrak{g} (t)$ and $\psi_2(x) = \inf_{t\in[x-1,x+1]}\mathfrak{u} (t)$.


\noindent In that case, the posterior distribution is strongly consistent at $\mathfrak{m}$ with respect to the weak topology on $\text{LR}_\text{joint}(\mathbbm{R}).$ As discussed above, it follows that the posterior distribution is strongly consistent at $(\mathfrak{f}, \mathfrak{g})$ with respect to the topology induced on $\text{LR}_\text{mix}(\mathbbm{R}).$ 
\end{theorem} 

We briefly remark upon assumptions B4, B5, and B6. The condition of assumption B4 is satisfied by most common densities. Assumption B5 introduces important moment conditions. Assumption B6 imposes regularity conditions. In the statement of assumption B6, the interval $[x-1,x+1]$ could be replaced with the interval $[x-c,x+c]$ for any $c > 0.$
 
\subsection{Computational implementation} \label{computation}

The posterior distribution is not available in closed form, but can be approximated via Markov chain Monte Carlo methods. We propose a slice-within-Gibbs sampler based on a truncated version of Sethuraman's stick-breaking representation of the Dirichlet process \citep{Sethuraman1994}. The implementation depends upon whether or not we intend to test the hypothesis $H_0: F = G$ against the alternative $H_1: F \leq_{\text{LR}}G$ using the spike-and-slab prior described in Subsection \ref{prior_description}. We briefly outline both strategies, leaving the full details to the supplementary material. 

If we are uninterested in the hypothesis testing problem, 
we assign a beta prior distribution distribution to $\theta.$ To sample $\theta$ from its full conditional distribution, it is helpful to introduce  binary latent variables that associate each observation $X_i$ with one of the two components in the mixture representation \eqref{fx_mixture}. The full conditional distribution for $\theta$ is then a beta distribution, while the full conditional distributions of the latent variables are Bernoulli distributions. To sample the atoms and weights of $P_1$ and $P_2$ from their full conditional distribution, we use the slice sampler described in Figure 8 of \citet{Neal2003}. We approximate the integral in the mixture representation of $f$ by numerical integration. As an alternative to numerical integration, we could introduce latent variables representing random truncation points.

More care is required if we are interested in the hypothesis testing problem. In that case, we assign a spike-and-slab prior to $\theta$ with a beta distribution for the slab. 
The Gibbs update for $\theta$ is slightly more involved than before but still relies upon standard techniques. When updating the atoms and weights of $P_1$ and $P_2,$ there are two cases. If $\theta < 1$, we sample the atoms and weights as described above. If $\theta=1$, the data are uninformative regarding $P_2.$ Thus, we sample the atoms and weights of $P_1$ using the slice sampling approach, but sample the atoms and weights of $P_2$ from their prior distribution. 
Sampling from the prior distribution can lead to values far outside the region where the posterior distribution is concentrated, negatively impacting convergence. To address this issue, we take advantage of the pseudo-prior framework proposed by \citet{Carlin1995}. 

\subsection{Simulation study} \label{simstudy}

We now evaluate the performance of the Bayesian nonparametric density estimation method introduced in Section \ref{prior_description} through a simulation study. Before presenting the results, we describe the simulation scenarios, the competing baseline method, the performance measures, and other details. 

The four simulation scenarios include two cases with a decreasing likelihood ratio and two cases with a constant likelihood ratio. To construct the data generating densities $\left(\mathfrak{f}, \mathfrak{g}\right) \in \text{LR}_\text{mix}(\mathbbm{R}),$ we specify $(\mathfrak{g}, \mathfrak{u}, \theta_0).$ By Theorem \ref{thm:fmixture_ac}, the pair $\left(\mathfrak{f}, \mathfrak{g}\right) \in \text{LR}_\text{mix}(\mathbbm{R})$ is uniquely determined through the mapping $T_1: (\mathfrak{g}, \mathfrak{u}, \theta_0) \mapsto (\mathfrak{f},\mathfrak{g}).$ 
The four scenarios are as follows: 

\paragraph{Scenario I:} In this scenario,
\begin{align*}
        \mathfrak{g}(\cdot) = 0.25 \cdot \phi_{1,\, 0.75^2}(\cdot) + 0.75 \cdot \phi_{-1,\, 0.75^2}(\cdot), \quad 
        \mathfrak{u}(\cdot) = \phi_{-0.5,\, 1}(\cdot), \quad \theta_0 = 0.1
\end{align*}
where $\phi_{\mu, \sigma^2}$ denotes the density of a Gaussian distribution with mean $\mu$ and variance $\sigma^2.$ Both $\mathfrak{f}$ and $\mathfrak{g}$ are unimodal with $\mathfrak{f}$ nearly symmetric. 
 
\paragraph{Scenario II:} In this scenario, $\mathfrak{f} = \mathfrak{g}$ with $\mathfrak{g}$ as defined in Scenario I. 

\paragraph{Scenario III:} In this scenario, 
\begin{align*}
        \mathfrak{g}(\cdot) &= 0.2 \cdot \phi_{1.5,\, 0.75^2}(\cdot) + 0.2 \cdot\phi_{-1.5,\, 0.75^2}(\cdot) + 0.6 \cdot \phi_{0.5,\, 1, \, 2}^{\rm skew}(\cdot), \\
        \mathfrak{u}(\cdot) &= 0.5 \cdot \phi_{1.5,\, 0.75^2}(\cdot) + 0.5 \cdot\phi_{-1.5,\, 0.75^2}(\cdot), \quad \theta_0 = 0.5
\end{align*}
where $\phi_{\xi,\, \omega, \, \alpha}^{\rm skew}$ is the density of a skew-normal distribution with location $\xi,$ scale $\omega,$ and shape $\alpha.$  Both $\mathfrak{f}$ and $\mathfrak{g}$ are bimodal and asymmetric.

\paragraph{Scenario IV:} In this scenario,  $\mathfrak{f} = \mathfrak{g}$ with $\mathfrak{g}$ as defined in Scenario III. \\ 

\noindent For each of the four scenarios, we consider a range of sample sizes. In particular, the sample sizes of the groups are balanced ($n = m$) with $n \in \{100, 250, 500\}$. For each scenario and sample size $n$, we simulate 100 data sets. 

To assess the benefits of incorporating the monotonicity constraint, we compare the results obtained with our approach to those obtained by modeling $(\mathfrak{f}, \mathfrak{g})$ as independent Dirichlet process mixtures of Gaussian densities. 

We recommend standardizing the data by subtracting the pooled sample mean and dividing by the pooled standard deviation. The standardized data fall within the range $(-4,4)$ with high probability, allowing us to define default prior distributions on the hyperparameters of the Dirichlet processes. We choose a normal inverse-gamma distribution as the base measure and set the precision parameter equal to one. The prior on $\theta$ is a spike-and-slab mixture with the spike at one and a Uniform$(0, 1)$ slab distribution. A detailed description of the prior specification is provided at the end of the section on computation in the supplementary material. 

The comparison between constrained and unconstrained approaches is based on the following posterior expected $L_1$ distances:
    \begin{align*}
    L_1(\mathfrak{f}) &= E(\|f - \mathfrak{f}\|_1 \mid X_1,\ldots, X_n, Y_1, \ldots, Y_m) \\
    L_1(\mathfrak{g}) &= E(\|g - \mathfrak{g}\|_1 \mid X_1,\ldots, X_n, Y_1, \ldots, Y_m) \\
    L_1(\mathfrak{f}, \mathfrak{g}) &= E(d_1\{(f,g), (\mathfrak{f},\mathfrak{g})\} \mid X_1,\ldots, X_n, Y_1, \ldots, Y_m) \\
    L_1(\mathfrak{f}/\mathfrak{g}) &= E(\|f/g - \mathfrak{f}/\mathfrak{g}\|_1 \mid X_1,\ldots, X_n, Y_1, \ldots, Y_m).
    \end{align*}
Because there is no guarantee that the $L_1$ distance between likelihood ratios is bounded, we compute $L_1(\mathfrak{f}/\mathfrak{g})$ over the range of the observed data. That is, we calculate
        \begin{align*}
    L_1(\mathfrak{f}/\mathfrak{g}) &
    & = E\left( \left. \int_{{\rm Range}(X_1,\ldots, X_n, Y_1, \ldots, Y_m)} \left|f(x)/g(x) - \mathfrak{f}(x)/\mathfrak{g}(x)\right|\, dx \right| X_1,\ldots, X_n, Y_1, \ldots, Y_m\right).
    \end{align*}

Table~\ref{tab:simulations} reports the posterior expected $L_1$ distances averaged over the simulated data sets and broken down by simulation scenario, sample size, and method. Each individual posterior expectation was approximated based on 1000 Markov Chain Monte Carlo draws obtained after 2000 iterations of burn-in and after thinning to retain every 10th iteration. 


The results reported in Table~\ref{tab:simulations} substantiate the value of the proposed density estimation method and affirm the benefits of incorporating the monotonicity constraint. For all scenarios and sample sizes, the constrained method provides better estimates of the densities $\left(\mathfrak{f},\mathfrak{g}\right)$ and the density ratio $\mathfrak{f}/\mathfrak{g}$ than the unconstrained method. The difference between the two methods is most pronounced when estimating the density ratio $\mathfrak{f}/\mathfrak{g}$. As we should expect, the average $L_1$ distances decrease as the sample size increases. Under Scenarios II and IV where $H_0: F  = G$ holds, the posterior probability of $H_0$ is large and increases with the sample size. Under Scenarios I and II where $H_1: F \leq_{\text{LR}} G$ holds, the posterior probability of $H_0$ is small and decreases with the sample size. 

In the supplementary material, we report other measures of performance. For example, we evaluate the posterior mean $E(f \mid X_1,\ldots, X_n, Y_1, \ldots, Y_m)$ as a point estimate of $\mathfrak{f}$ by looking at the $L_2$ distance $L_2(\mathfrak{f}) = \|E(f \mid X_1,\ldots, X_n, Y_1, \ldots, Y_m) - \mathfrak{f}\|_2$ averaged over the simulated data sets. The main conclusion is the same: In each of the scenarios considered, incorporating the monotonicity constraint leads to improved estimates. 


\begin{table} 
\centering  
\caption{Average posterior expected $L_1$ distances and posterior probability of $H_0$.} \label{tab:simulations}
\begin{tabular}{c|c|cccc|cccc|c}
         &   & \multicolumn{4}{c|}{LR order constrained model}  & \multicolumn{4}{c|}{Unconstrained model}  & Prob. of\\  
S & $n$ & $\bar L_1(\mathfrak{f})$ & $\bar L_1(\mathfrak{g})$  & $\bar L_1(\mathfrak{f}, \mathfrak{g}) $  & $\bar L_1(\mathfrak{f}/\mathfrak{g}) $  & $\bar L_1(\mathfrak{f})$ & $\bar L_1(\mathfrak{g})$  & $\bar L_1(\mathfrak{f}, \mathfrak{g}) $  & $\bar L_1(\mathfrak{f}/\mathfrak{g})$ & $H_0:f=g$ \\ 
  \hline
  I & 100 & 0.012 & 0.015 & 0.027 & 2.184 & 0.012 & 0.017 & 0.029 & 12.680 & 0.002 \\ 
  I & 250 & 0.008 & 0.010 & 0.017 & 1.791 & 0.008 & 0.011 & 0.019 & 5.028 & 0.001 \\ 
  I & 500 & 0.006 & 0.007 & 0.013 & 1.504 & 0.006 & 0.008 & 0.014 & 2.711 & 0.001 \\ 
  \hline
  II & 100 & 0.011 & 0.011 & 0.022 & 0.005 & 0.016 & 0.015 & 0.031 & 0.448 & 0.816 \\ 
  II & 250 & 0.007 & 0.007 & 0.014 & 0.004 & 0.010 & 0.010 & 0.020 & 0.328 & 0.850 \\ 
  II & 500 & 0.005 & 0.005 & 0.010 & 0.000 & 0.007 & 0.007 & 0.014 & 0.279 & 0.913 \\ 
  \hline
  III & 100 & 0.019 & 0.018 & 0.037 & 0.284 & 0.024 & 0.025 & 0.050 & 1.030 & 0.062 \\ 
  III & 250 & 0.011 & 0.011 & 0.021 & 0.199 & 0.014 & 0.015 & 0.029 & 0.475 & 0.002 \\ 
  III & 500 & 0.008 & 0.008 & 0.015 & 0.168 & 0.010 & 0.009 & 0.019 & 0.328 & 0.001 \\ 
  \hline
  IV & 100 & 0.018 & 0.018 & 0.035 & 0.008 & 0.026 & 0.026 & 0.052 & 0.524 & 0.821 \\ 
  IV & 250 & 0.011 & 0.011 & 0.022 & 0.008 & 0.015 & 0.016 & 0.031 & 0.448 & 0.855 \\ 
  IV & 500 & 0.008 & 0.008 & 0.015 & 0.003 & 0.010 & 0.010 & 0.021 & 0.343 & 0.896 \\ 
  \end{tabular}
\end{table}


\subsection{Analysis of C-reactive protein data} \label{subsec:CHOP}

Following \citet{Westling2023}, we apply our method to data from a study by \citet{Downes2020}. The study enrolled children who presented symptoms of systemic inflammatory response syndrome (SIRS) at the pediatric intensive care unit of the Children's Hospital of Philadelphia between August 10, 2012 and June 9, 2016. Rapid treatment with broad-spectrum antibiotics is associated with reduced morbidity and mortality in SIRS patients with a bacterial infection. However, many patients with SIRS do not have a bacterial infection, and overuse of broad-spectrum antibiotics presents risks for the individual patient and the broader community. \citet{Downes2020} investigated whether the commonly available biomarker C-reactive protein can be used to guide antibiotic discontinuation in pediatric SIRS patients. 

The data set analyzed by \citet{Westling2023} includes pairs $\{Z_i, D_i\}$ of the C-reactive protein measurement $Z_i,$ in milligrams per deciliter, and the infection status $D_i$ of each patient for which this information was recorded. Although some patients had multiple SIRS episodes, these are treated as independent, because all such episodes occurred at least 30 days apart. Due to the inaccuracy of the assay for low C-reactive protein concentrations, values below $0.5$ were rounded up to equal $0.5.$ In contrast to \citet{Westling2023}, we exclude these measurements.  In the end, we are left with 454 C-reactive protein measurements. Of these, 194 come from patients with a bacterial infection, while 260 come from patients without a bacterial infection.

\begin{figure}[htbp]
\centerline{\includegraphics[width=6in]{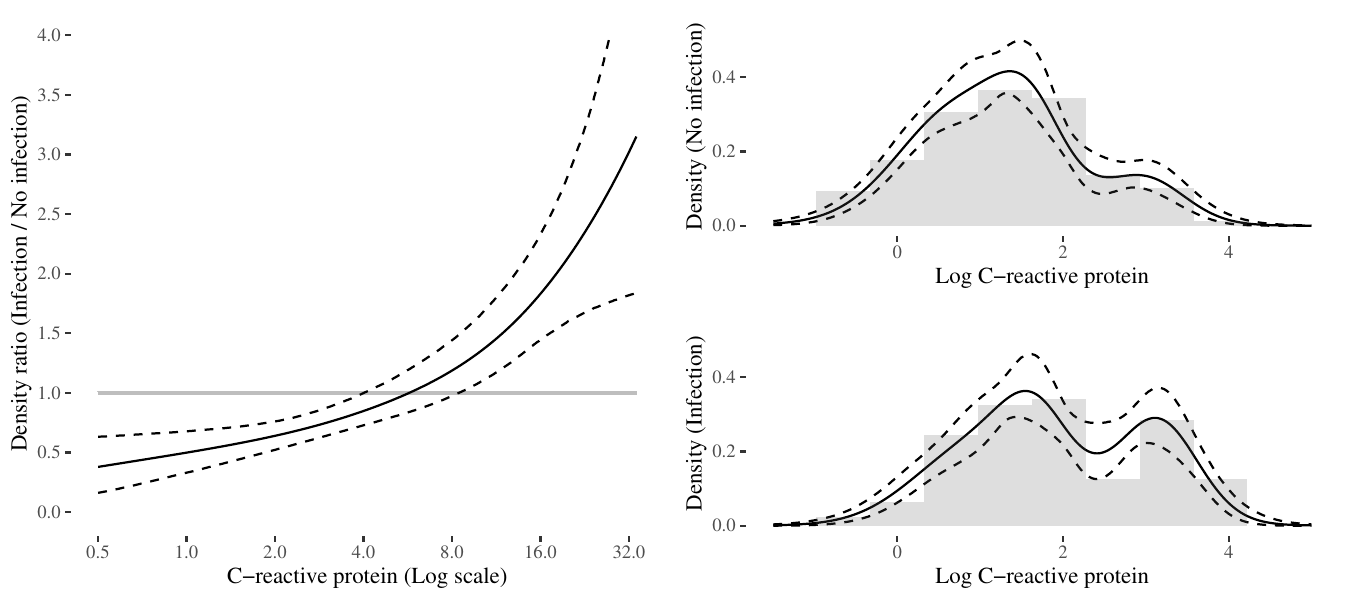}}
\caption{Left: The posterior mean of the density ratio $g/f$ (solid) along with 95\% pointwise credible intervals (dashed). The logarithmic scale is meant to ease comparison with Fig. 4 of \citet{Westling2023}. Right: Density estimates (solid) and 95\% pointwise credible intervals (dashed) plotted on top of histograms calculated from the log C-reactive protein measurements.}
\label{fig:CHOP_plot}
\end{figure}

High concentrations of C-reactive protein are associated with bacterial infection in pediatric SIRS patients \citep{Downes2017}. Thus, there is reason to expect that a likelihood ratio ordering holds. Let $f$ and $g$ denote the densities of the log C-reactive protein measurements conditional on non-infection and infection, respectively. We conduct posterior inference for $f$ and $g$ using the prior distribution described earlier in this section. Because the likelihood ratio order is preserved under monotonic transformations \citep[Theorem 1.C.8]{Shaked2007}, we can make inferences about the densities of the untransformed C-reactive protein measurements without issue through a change of variables. Details on the choice of prior hyperparameters and the computational implementation can be found in the supplementary material.

Figure \ref{fig:CHOP_plot} presents the results of our analysis, allowing comparison with \citet{Westling2023} and \citet{Downes2020}. The emphasis in \citet{Westling2023} is on the ratio 
\begin{align*} 
    \frac{g(z)}{f(z)} &= \frac{\text{pr}(D_i=1 \mid Z_i=z)/\text{pr}(D_i=1)}{\text{pr}(D_i=0 \mid Z_i=z)/\text{pr}(D_i=0)}.
\end{align*} 
As this equation makes clear, the density ratio can be interpreted as an odds ratio. In particular, $g(z)/f(z)$ is the multiplicative factor by which we adjust the prior odds ${\text{pr}(D=1)/\text{pr}(D=0)}$ upon observing $Z=z.$ The left panel of Fig. \ref{fig:CHOP_plot} displays the posterior mean of $g/f$ along with 95\% pointwise credible intervals. 
The posterior mean and credible intervals are qualitatively similar to the estimates and intervals presented in Fig. 4 of \citet{Westling2023}, except the former are smooth while the latter are piecewise constant. Our credible intervals are narrower than the corresponding confidence intervals of \citet{Westling2023}. This is unsurprising, because we make stronger assumptions on the distributions of interest. The results indicate that C-reactive protein measurements below $4$ are associated with decreased odds of bacterial infection, while C-reactive protein measurements above $8$ are associated with increased odds of bacterial infection. C-reactive protein measurements between $4$ and $8$ are not clearly associated with either increased or decreased odds of bacterial infection. These conclusions align with the algorithm described in \citet{Downes2017} which suggests discontinuing antibiotics in pediatric SIRS patients with C-reactive protein measurements below $4,$ provided that those patients also meet certain other criteria. The right panel of Fig. \ref{fig:CHOP_plot} displays the density estimates along with 95\% pointwise credible intervals. To our knowledge, there is no other method for density estimation under a likelihood ratio order constraint that also allows for uncertainty quantification. 
That there were no Markov chain Monte Carlo iterations with $\theta = 1$ is strong evidence against the null hypothesis $H_0: F = G.$

The C-reactive protein data were obtained through a data sharing agreement with the Children's Hospital of Philadelphia that does not permit the data to be made public. 
Therefore, we present an additional application in the supplementary material based on the publicly available survival data analyzed in \citet{Hoel1972} and \citet{Carolan2005}.

\section{Connections to other statistical problems}

Previous work on statistical methods for likelihood ratio ordered distributions has highlighted connections with seemingly unrelated statistical tasks. 
\citet{Dykstra1995} pointed out how their methods can be used to compare trends in nonhomogenous Poisson processes, while \citet{Yu2017} applied their density estimation method to shape-constrained estimation of receiver operating characteristic curves. The connection between the likelihood ratio order and selection models is also worthy of comment. 

Accounting for selection bias is a central problem in statistics. The problem is often formulated as follows: We are interested in the distribution of a random variable $Z$ with density $g,$ but we only observe $Z$ with probability $w(Z).$ The conditional density of $Z,$ given that it is observed, is $f(z) \propto w(z) g(z),$ and the function $w$ is called the biasing or weight function. 

In many applications, we expect the weight function $w$ to be monotone. 
\citet{Sun1997} gave examples from quality control, sample surveys, astronomy, wildlife biology, and oil discovery. In meta-analysis, we expect that studies with small p-values have a higher probability of being published than those with large p-values \citep{Bayarri1987, Iyengar1988, Silliman1997, Silliman1997a}. A monotone weight function $w(x) \propto x$ gives rise to so-called size-biased or length-biased data \citep{Cox1969, Patil1978}. 


If the weight function $w$ is monotone, the ratio $f/g \propto w$ is likewise monotone. In that case, the mixture representations introduced in this article can be brought to bear on the problem of selection bias. A semi-parametric approach along the lines of \citet{Sun1997} or \citet{Lee2001} would be a natural direction for future work. 

\section*{Acknowledgement}

We would like to acknowledge Craig Boge, Susan E. Coffin, and Kevin J. Downes of the Children's Hospital of Philadelphia for providing the C-reactive protein data.

\bibliographystyle{apalike}
\bibliography{references}

\clearpage

\begin{center}
{\huge \textbf{Supplementary material}}
\end{center}
\appendix


\section*{Proof of Theorem~\ref{thm:fmixture_finite} } \label{subsec:f_finite}

\subsection*{The backward direction}

Suppose that there exists $\theta \in [0,1]$ and a mass function $u$ with $\sum_{x_j \in \mathcal{I} \setminus \{x_d\}} u(x_j) = 1$ 
such that \eqref{fmixture_finite} holds. Our goal is to show that the ratio $r = f/g$ is non-increasing on $\mathcal{I}$. The ratio can be expressed as
\begin{align} \label{discrete_ratio}
    r(x_i) &=    \theta  + (1-\theta) \,  \sum_{\substack{x_j \in \mathcal{I} \setminus \{x_d\} \\  x_j \ge x_i }} \frac{u(x_j)}{G(x_j)}, \quad x_i \in \mathcal{I}.
\end{align} 
As $x_i$ increases, $r(x_i)$ is the sum of fewer non-negative terms. Thus, the ratio $r$ is non-increasing.

\subsection*{The forward direction}

Suppose the ratio $r=f/g$ is non-increasing on $\mathcal{I}$. Our goal is to find $\theta \in [0,1]$ and a mass function $u$ with $\sum_{x_j \in \mathcal{I} \setminus \{x_d\}} u(x_j) = 1$ such that \eqref{fmixture_finite} holds. 

Let $\theta = r(x_d) = \min_{x_i \in \mathcal{I}} r(x_i).$ One can verify that $\theta \in [0,1]$ with $\theta = 1$ if only if $f=g$ on $\mathcal{I}.$ 

We turn our attention to finding a mass function $u$ with $\sum_{x_j \in \mathcal{I} \setminus \{x_d\}} u(x_j) = 1$ such that \eqref{fmixture_finite} holds. There are two cases: $\theta = 1$ and $\theta \in [0,1).$ When $\theta = 1,$ we know $f=g$ on $\mathcal{I}.$ In that case, \eqref{fmixture_finite} holds for any mass function $u$ with $\sum_{x_j \in \mathcal{I} \setminus \{x_d\}} u(x_j) = 1.$ When $\theta \in [0,1),$ we set $ u(x_j) = {G(x_j)}\left\{r(x_j) - r(x_{j+1})\right\}/(1-\theta)$ for $x_j \in \mathcal{I}\setminus\{x_d\}.$

First, we verify that $u$ is a mass function. It is non-negative because $G$ is a distribution function, $r$ is non-increasing, and $\theta \in [0,1).$ We also have that 
\begin{align*}
    (1-\theta)  \sum_{x_j \in \mathcal{I} \setminus \{x_d\}} u(x_j) 
    &= \sum_{x_j \in \mathcal{I} \setminus \{x_d\}} G(x_j) \left\{ r(x_j) - r(x_{j+1}) \right\}  \\
    &=  - \underbrace{ G(x_{d-1}) \, \theta}_{\theta - f(x_d)} + \underbrace{G(x_1) r(x_1)}_{f(x_1)} +   \underbrace{\sum_{x_j \in \mathcal{I} \setminus \{x_1, x_d\}}  r(x_j) \left\{ 
    G(x_j) - 
    G(x_{j-1}) \right\}}_{\sum_{x_j \in \mathcal{I} \setminus \{x_1, x_d\}} f(x_j)}  \\
    &=  \sum_{x_j \in \mathcal{I}} f(x_j)  - \theta \\
    &= 1 - \theta,
\end{align*}
implying that $\sum_{x_j \in \mathcal{I} \setminus \{x_d\}} u(x_j)  = 1.$

Now we verify that \eqref{fmixture_finite} holds. 
For $x_i \in \mathcal{I},$ we have
\begin{align*}
     \theta g(x_i) +  (1-\theta) \sum_{x_j \in \mathcal{I} \setminus \{ x_d \}} u(x_j) g^{x_j}(x_i) &= \theta g(x_i) +  g(x_i) \sum_{\substack{x_j \in \mathcal{I} \setminus \{x_d\} \\ x_j \ge x_i } } \{r(x_j) - r(x_{j+1})\} \\
     &= \theta g(x_i) + g(x_i)\{ r(x_i) - \theta \} \\ 
     &= f(x_i).
\end{align*}

\subsection*{Uniqueness}

Suppose that the equivalent conditions of Theorem \ref{thm:fmixture_finite} hold. That $\theta = r(x_d)$ follows from evaluating \eqref{discrete_ratio} at $x_d.$ Now suppose $\theta \in [0,1).$ Our goal is to show that $u$ is uniquely determined. Toward that end, suppose there exists another probability mass function $\tilde{u}$ with $\sum_{x_j \in \mathcal{I} \setminus \{x_d\}} \tilde{u}(x_j) = 1$ such that 
\begin{align*}
    f(x_i) = \theta g(x_i) + (1-\theta) \sum_{x_j \in \mathcal{I} \setminus \{ x_d \}}  g^{x_j}(x_i) \tilde{u}(x_j), \quad x_i \in \mathcal{I}.
\end{align*}
Then we have
\begin{align*}
0 &= f(x_i) - f(x_i) =  (1-\theta)  g(x_i)  \sum_{\substack{ x_j \in \mathcal{I} \setminus \{x_d\} \\ x_j \ge x_i } } \frac{u(x_j) - \tilde{u}(x_j)} {G(x_j)}, \quad x_i \in \mathcal{I}.
\end{align*}
Because $\theta \in [0,1)$ and $g > 0$ on $\mathcal{I},$ the sum on the right hand side is equal to zero for each $x_i \in \mathcal{I}.$ By taking $x_i = x_{d-1},$ we see that $u(x_{d-1}) = u(x_{d-1}).$ By taking $x_i = x_{d-2},$ we then see that $u(x_{d-2}) = u(x_{d-2}).$ Continuing in this way, we see that $u = \tilde{u}$ on $\mathcal{I} \setminus \{x_d\}$ and conclude that $u$ is uniquely determined. 
 


\section*{Proof of Theorem~\ref{thm:gmixture_finite}} \label{subsec:g_finite}

The proof is omitted because it is similar to that of Theorem~\ref{thm:fmixture_finite}.

\section*{Proof of Theorem \ref{thm:fmixture_ac}}

\subsection*{The backward direction}

Suppose that there exists $\theta \in [0,1]$ and an absolutely continuous distribution function $U$ with $U(a+)=0$ and $U(b-)=1$ such that \eqref{fx_mixture} holds. Our goal is to show that the ratio $r = f/g$ is non-increasing and locally absolutely continuous on $I.$

First, we show that $r$ is non-increasing on $I.$ The ratio can be expressed as 
\begin{align} \label{ratio}
    r(x) &= \theta + (1-\theta) \int_{x}^b \frac{u(s)}{G(s)}\, ds, \quad x\in I.
\end{align}
The integrand is non-negative and the region of integration shrinks as $x$ increases. Thus, the ratio $r$ is non-increasing. 

 Next, we show that $r \in AC_\text{loc}(I).$ Let $J = [c,d] \subset (a,b)$ be given. When $\theta = 1,$ the ratio $r$ is constant and thus absolutely continuous on $J.$ Now suppose that $\theta \in [0,1)$ and define 
 \begin{align*}
     h(x) = \int_c^x \frac{u(s)}{G(s)}\, ds, \quad x\in J.
 \end{align*} 
 The function $h$ is absolutely continuous on $J$ \citep[Theorem 3.16]{Leoni2017}.  For all $x \in J,$ we have that 
 \begin{align*}
     \int_c^b \frac{u(s)}{G(s)}\, ds &= \int_c^x \frac{u(s)}{G(s)}\, ds + \int_x^b \frac{u(s)}{G(s)}\, ds \\ 
     &= h(x) + \frac{r(x)-\theta}{1-\theta}.
 \end{align*} After rearranging, we see that $r(x) =\kappa_1 + \kappa_2 \cdot h(x)$ for constants {$\kappa_1 = \theta + (1-\theta) \int_c^b \frac{u(s)}{G(s)}\, ds$} and $\kappa_2 = -(1-\theta).$ Elementary properties of absolutely continuous functions allow us to conclude that $r$ is absolutely continuous on $J$ \citep[Exercise 3.7]{Leoni2017}. Thus,  $r \in AC_\text{loc}(I).$

\subsection*{The forward direction}

Suppose the ratio $r = f/g$ is non-increasing and locally absolutely continuous on $I.$ Our goal is to find $\theta \in [0,1]$ and an absolutely continuous distribution function $U$  with $U(a+)=0$ and $U(b-)=1$ such that \eqref{fx_mixture} holds.


Let $\theta = r(b-) = \inf_{x\in I} r(x).$ The limit exists and is equal to the infimum because the ratio $r$ is non-increasing and bounded below on $I.$ 
One can verify that $\theta \in [0,1]$ with $\theta=1$ if and only if $f=g$ everywhere on $I.$

We turn our attention to finding an absolutely continuous distribution function $U$  with $U(a+)=0$ and $U(b-)=1$ such that $\eqref{fx_mixture}$ holds. There are two cases:  $\theta=1$ and $\theta \in [0,1).$ When $\theta = 1,$ we know $f=g$ everywhere on $I.$ In that case, $\eqref{fx_mixture}$ holds for any absolutely continuous distribution function $U$ with $U(a+)=0$ and $U(b-)=1.$  When $\theta \in [0,1),$ we define $R=F/G$ and set $U(x) = G(x)\left\{R(x) - r(x)\right\}/(1-\theta)$ for $ x \in I.$

First, we verify that $U$ is an absolutely continuous distribution function: 

\noindent \textit{Absolute continuity.}
Rearranging the expression for $U,$ we have $U(x) = \left\{F(x) - G(x) r(x)\right\}/\\ (1-\theta)$ for $x \in I.$
That $U \in AC_\text{loc}(I)$ follows from \citet[Exercise 3.7]{Leoni2017} and the local absolute continuity of $F, G,$ and $r.$

\noindent \textit{Monotonicity.} 
Let $c, d \in I$ with $c < d.$ 
The function $U$ is differentiable almost everywhere on $J = [c,d],$ with $U'(x) = G(x)\{-r'(x)\}/(1-\theta)$ for almost all $x\in J.$ The ratio $r$ is non-increasing, implying that its derivative $r'(x) \leq 0$ at $x \in J$ for which it exists. Thus, $U'(x) \geq 0$ for almost all $x \in J.$ 
By Proposition 2 of \citet{Koliha2009}, $U$ is non-decreasing on $J.$ In particular, we have $U(d) - U(c) \geq 0,$ which allows us to conclude that $U$ is non-decreasing on $I.$

\noindent \textit{Non-negativity.} Let $x_0 \in I$ be given. If $R(x_0) - r(x_0) \geq 0,$ it immediately follows that $U(x_0) \geq 0.$ For all $x \in (a,x_0],$ we have $f(x)g(x_0) \geq f(x_0)g(x).$ Integrating both sides from $x=a$ to $x=x_0$ yields the desired inequality.


\noindent \textit{Limits.} Verifying that $U(b-) = 1$ is straightforward. Now we evaluate the limit $U(a+).$ It follows from the non-negativity of $U$ that $F(x) \geq G(x) r(x) \geq 0$ for $x \in I.$ This inequality implies $\lim_{x \downarrow a} G(x)r(x) = 0,$ from which we can easily conclude $U(a+)=0.$

Now we verify that \eqref{fx_mixture} holds. Let $x_0 \in I$ be given. For all $c,d \in I$ with $c < d,$ we have 
\begin{align*}
    \theta g(x_0) + (1-\theta) \int_c^d g^s(x_0) \, U'(s) \,  ds 
    &= g(x_0)\left\{\theta - \int_{x_0}^d r'(s)\, ds\right\} \\ 
    &= g(x_0)\left\{\theta + r(x_0) - r(d)\right\}.
\end{align*}
The first equality comes from substituting in known expressions for $g^s(x_0)$ and $U'(x_0)$ and then simplifying, while the second equality follows from the local absolute continuity of $r$ and a version of the fundamental theorem of calculus \citep[Theorem 3.20]{Leoni2017}. By taking the limit of both sides as $c \downarrow a$ and $d \uparrow b,$ we conclude that \eqref{fx_mixture} holds.

\subsection*{Uniqueness}

Suppose the equivalent conditions of Theorem \ref{thm:fmixture_ac} hold. That $\theta=r(b-)$ follows from taking the limit as $x\uparrow b$ of both sides of  \eqref{ratio}. Now suppose $\theta \in [0,1).$ Our goal is to show that $U$ is uniquely determined with $U(x) = G(x)\left\{R(x) - r(x)\right\}/(1-\theta)$ for all  $x \in I.$ Let $x_0 \in I.$ Making use of \eqref{fx_mixture}, we have
\begin{align*}
F(x_0) &= \int_a^{x_0} f(t) \, dt \\ 
&= \theta G(x_0) + (1 - \theta) \int_a^{x_0} \left\{ \int_a^b g^s(t) \, U'(s) \,  ds \right\} dt \\ 
&= \theta G(x_0) + (1 - \theta) \int_a^b \left\{ \int_a^{x_0} g^s(t) \, U'(s) \,  dt\right\} ds \\ 
&= \theta G(x_0) + (1 - \theta) \int_a^b \left\{ \int_a^{s \wedge x_0} g(t) \, dt\right\} \frac{U'(s)}{G(s)} ds \\ 
&= \theta G(x_0) + (1 - \theta) \int_a^b G(s \wedge x_0) \frac{U'(s)}{G(s)} \,  ds.
\end{align*}
Exchanging the order of integration is allowed by Fubini's theorem. The hypotheses of that theorem are satisfied because $g^s(t) \, U'(s)$ is a joint density in $(s,t)$ on $I \times I.$ The integrand of the last integral can be expressed as
\begin{align*}
     G(s \wedge x_0) \frac{U'(s)}{G(s)} & = 
     \begin{cases}
     U'(s), \quad &s \leq x_0 \\ 
     G(x_0)\frac{U'(s)}{G(s)}, \quad &s > x_0
     \end{cases}
\end{align*} so that 
\begin{align*}
F(x_0) &= \theta G(x_0) + (1 - \theta) 
\left\{U(x_0) + G(x_0) \int_{x_0}^b \frac{U'(s)}{G(s)}\,  ds \right\}
\end{align*} From the proof of the backward direction, we know that 
\begin{align*}
    \int_{x_0}^b \frac{U'(s)}{G(s)}\,  ds &= \frac{r(x_0)-\theta}{1-\theta}.
\end{align*} Substituting this into the previous equation and rearranging yields the desired result. 

\section*{Proof of Proposition \ref{prop:suff_cond_fx}}

The proposition follows from Exercise 3.48 and Corollary 1.25 of \citet{Leoni2017}. Suppose $f/g$ is continuous, non-increasing, and differentiable except at countably many points on $I.$ Let a compact interval $J \subset I$ be given. By Exercise 3.48, we only need to show that $(f/g)'$ is Lebesgue integrable over $J$ in order to conclude that $f/g$ is absolutely continuous on $J.$ Since $f/g$ is monotone, the integrability of $(f/g)'$ over $J$ follows immediately from Corollary 1.25.

\section*{Proof of Theorem \ref{thm:gmixture_ac}}

The proof is omitted because it is similar to that of Theorem~\ref{thm:fmixture_ac}.

\section*{Proof of Proposition \ref{prop:suff_cond_gx}}

The proof is essentially identical to that of Proposition \ref{prop:suff_cond_fx}. Simply replace the non-increasing ratio $f/g$ with the non-decreasing ratio $g/f.$

\section*{Proof of Theorem \ref{thm:support}}

To prove Theorem \ref{thm:support}, we need two lemmas. The first concerns continuity of the map $T_1: (g, u, \theta) \mapsto (f,g),$ while the second concerns the $L_1$-support of Dirichlet process location-scale mixtures of Gaussian distributions. Theorem \ref{thm:support} follows easily from the lemmas. 

Denote the space of probability density functions on $\mathbbm{R}$ by $\Delta,$ and denote the subset of strictly positive density functions by $\Delta_{+}.$ As in the main text, the notation $\|\cdot\|_1$ indicates the $L_1$ norm. We endow the domain $\Delta_{+} \times \Delta \times [0,1]$ and codomain $\Delta \times \Delta_{+}$ of the map $T_1$ with metrics $d_1\{(h,\ell), (\tilde{h}, \tilde{\ell})\} = \|h - \tilde{h}\|_1 + \|\ell - \tilde{\ell}\|_1$ for $(h,\ell), (\tilde{h}, \tilde{\ell}) \in \Delta \times \Delta_{+}$ and $d_2\{(h,\ell, \theta), (\tilde{h}, \tilde{\ell}, \tilde{\theta})\} = d_1\{(h,\ell), (\tilde{h}, \tilde{\ell})\} + |\theta - \tilde{\theta}|$ for $(h,\ell, \theta), (\tilde{h}, \tilde{\ell}, \tilde{\theta}) \in \Delta_{+} \times \Delta \times [0,1].$



\begin{lemma*}[Continuity] 
The map $T_1$ is continuous on $\Delta_{+} \times \Delta \times [0,1]$ with respect to the given metrics.
\end{lemma*}
\begin{proof} Let $(\mathfrak{g},\mathfrak{u},\theta_0)\in \Delta_{+} \times \Delta \times [0,1]$ and $\epsilon > 0$ be given. We need to find $\delta>0$ such that $d_2\{(g,u,\theta),(\mathfrak{g},\mathfrak{u},\theta_0)\}<\delta$ implies $d_1\{(f,g),(\mathfrak{f},\mathfrak{g})\}<\epsilon$. To do so, we start by noticing that
\begin{align}
d_{1}\{(f,g),(\mathfrak{f},\mathfrak{g})\} 
=    &  \|f-\mathfrak{f}\|_{1}+\|g-\mathfrak{g}\|_{1}\notag\\
=    &  \int_{-\infty}^{\infty}\left|f(x)-\mathfrak{f}(x)\right|\,dx+\|g-\mathfrak{g}\|_{1}\notag\\
=    &  \int_{-\infty}^{\infty}\left|\theta g(x)+(1-\theta)\int_{-\infty}^{\infty}g^{s}(x)u(s)\,ds-\theta_{0}\mathfrak{g}(x)-\left(1-\theta_{0}\right)\int_{-\infty}^{\infty}\mathfrak{g}^{s}(x)\mathfrak{u}(s)\,ds\right|\,dx\notag\\
     &  +\|g-\mathfrak{g}\|_{1}\notag\\
\leq &  \|\theta g-\theta_{0}\mathfrak{g}\|_{1}+\int_{-\infty}^{\infty}\left|(1-\theta)\int_{-\infty}^{\infty}g^{s}(x)u(s)\,ds-(1-\theta_{0})\int_{-\infty}^{\infty}\mathfrak{g}^{s}(x)\mathfrak{u}(s)\,ds\right|\,dx\notag\\
     &  +\|g-\mathfrak{g}\|_{1}\notag\\
=    &  \|\theta g-\theta_{0}\mathfrak{g}\pm\theta\mathfrak{g}\|_{1}\notag\\
     &  +\int_{-\infty}^{\infty}\left|(1-\theta)\int_{-\infty}^{\infty}g^{s}(x)u(s)\,ds\right.\notag\\
     &  - \left.(1-\theta_{0})\int_{-\infty}^{\infty}\mathfrak{g}^{s}(x)\mathfrak{u}(s)\,ds\pm(1-\theta)\int_{-\infty}^{\infty}\mathfrak{g}^{s}(x)\mathfrak{u}(s)\,ds\right|\,dx\notag\\
     &  +\|g-\mathfrak{g}\|_{1}\notag\\
\leq &  \|\theta g-\theta\mathfrak{g}\|_{1}+\|\theta\mathfrak{g}-\theta_{0}\mathfrak{g}\|_{1}\notag\\
     &  +\int_{-\infty}^{\infty}\left|(1-\theta)\int_{-\infty}^{\infty}g^{s}(x)u(s)\,ds-(1-\theta)\int_{-\infty}^{\infty}\mathfrak{g}^{s}(x)\mathfrak{u}(s)\,ds\right|\,dx\notag\\
     &  +\int_{-\infty}^{\infty}\left|(1-\theta)\int_{-\infty}^{\infty}\mathfrak{g}^{s}(x)\mathfrak{u}(s)\,ds-(1-\theta_{0})\int_{-\infty}^{\infty}\mathfrak{g}^{s}(x)\mathfrak{u}(s)\,ds\right|\,dx\notag\\
     &   +\|g-\mathfrak{g}\|_{1}\notag\\
=    &   \theta\|g-\mathfrak{g}\|_{1}+\left|\theta-\theta_{0}\right|\|\mathfrak{g}\|_{1}\notag\\
     &   +(1-\theta)\int_{-\infty}^{\infty}\left|\int_{-\infty}^{\infty}g^{s}(x)u(s)\,ds-\int_{-\infty}^{\infty}\mathfrak{g}^{s}(x)\mathfrak{u}(s)\,ds\right|\,dx\notag\\
     &   +\left|\theta-\theta_{0}\right|\int_{-\infty}^{\infty}\left|\int_{-\infty}^{\infty}\mathfrak{g}^{s}(x)\mathfrak{u}(s)\,ds\right|\,dx\notag\\
     &  +\|g-\mathfrak{g}\|_{1}\notag\\
\leq &  \|g-\mathfrak{g}\|_{1}+\left|\theta-\theta_{0}\right|\notag\\
     &  +\int_{-\infty}^{\infty}\left|\int_{-\infty}^{\infty}g^{s}(x)u(s)\,ds-\int_{-\infty}^{\infty}\mathfrak{g}^{s}(x)\mathfrak{u}(s)\,ds\right|\,dx\notag\\
    &   +\left|\theta-\theta_{0}\right| +\|g-\mathfrak{g}\|_{1}\notag\\
=   &   2\|g-\mathfrak{g}\|_{1}+2\left|\theta-\theta_{0}\right|+\int_{-\infty}^{\infty}\left|\int_{-\infty}^{\infty}g^{s}(x)u(s)\,ds-\int_{-\infty}^{\infty}\mathfrak{g}^{s}(x)\mathfrak{u}(s)\,ds\right|\,dx. \label{eq:LemmaContIneq1}
\end{align}
We now focus on finding an upper bound for the integral in the last line of the inequality \eqref{eq:LemmaContIneq1}. For any $s_0\in \mathbb{R}$,
\begin{align}
&   \hspace{-25mm}
 \int_{-\infty}^{\infty}\left|\int_{-\infty}^{\infty}g^{s}(x)u(s)\,ds-\int_{-\infty}^{\infty}\mathfrak{g}^{s}(x)\mathfrak{u}(s)\,ds\right|\,dx \notag\\
& \hspace{-15mm}  = 
   \int_{-\infty}^{\infty}\left|\int_{-\infty}^{s_0}g^{s}(x)u(s)\,ds-\int_{-\infty}^{s_0}\mathfrak{g}^{s}(x)\mathfrak{u}(s)\,ds \right.    \notag\\ 
& \hspace{-11mm}
\left. + \int_{s_0}^{\infty}g^{s}(x)u(s)\,ds-\int_{s_0}^{\infty}\mathfrak{g}^{s}(x)\mathfrak{u}(s)\,ds\right|\,dx
   \notag\\ 
& \hspace{-15mm}  \leq  
   \int_{-\infty}^{\infty}\left|\int_{-\infty}^{s_0}g^{s}(x)u(s)\,ds-\int_{-\infty}^{s_0}\mathfrak{g}^{s}(x)\mathfrak{u}(s)\,ds \right|\,dx   \notag\\ 
& \hspace{-11mm}
 +
\int_{-\infty}^{\infty}\left| \int_{s_0}^{\infty}g^{s}(x)u(s)\,ds-\int_{s_0}^{\infty}\mathfrak{g}^{s}(x)\mathfrak{u}(s)\,ds\right|\,dx.
   \label{eq:LemmaContIneq2} 
\end{align}
Denote the distribution functions corresponding to $\mathfrak{g}$ and $\mathfrak{u}$ by $\mathfrak{G}$ and $\mathfrak{U},$ respectively. For the first integral in the last line of the inequality \eqref{eq:LemmaContIneq2}, we have that
\begin{align}
 &  \hspace{-25mm}  \int_{-\infty}^{\infty}\left|\int_{-\infty}^{s_{0}}g^{s}(x)u(s)\,ds-\int_{-\infty}^{s_{0}}\mathfrak{g}^{s}(x)\mathfrak{u}(s)\,ds\right|\,dx\notag \\
\leq & \int_{-\infty}^{\infty}\int_{-\infty}^{s_{0}}g^{s}(x)u(s)\,ds\,dx+\int_{-\infty}^{\infty}\int_{-\infty}^{s_{0}}\mathfrak{g}^{s}(x)\mathfrak{u}(s)\,ds\,dx\notag \\
=    & \int_{-\infty}^{s_{0}}\int_{-\infty}^{\infty}g^{s}(x)u(s)\,dx\,ds+\int_{-\infty}^{s_{0}}\int_{-\infty}^{\infty}\mathfrak{g}^{s}(x)\mathfrak{u}(s)\,dx\,ds\notag \\
=    & \int_{-\infty}^{s_{0}}\int_{-\infty}^{\infty}g^{s}(x)\,dx\,u(s)\,ds+\int_{-\infty}^{s_{0}}\int_{-\infty}^{\infty}\mathfrak{g}^{s}(x)\,dx\,\mathfrak{u}(s)\,ds\notag \\
=    &\, U(s_{0})+\mathfrak{U}(s_{0})\notag \\
=    & \left\{U(s_{0})-\mathfrak{U}(s_{0})\right\}+2\mathfrak{U}(s_{0})\notag \\
\leq & \sup_{s\in\mathbb{R}}\left\{U(s)-\mathfrak{U}(s)\right\}+2\mathfrak{U}(s_{0})\notag \\
\leq &\, \|u-\mathfrak{u}\|_{1}+2\mathfrak{U}(s_{0}). \label{eq:LemmaContIneq2a}
\end{align}
For the second integral in the last line of the inequality  \eqref{eq:LemmaContIneq2}, we have that
\begin{align}
     &  \int_{-\infty}^{\infty}\left|\int_{s_{0}}^{\infty}g^{s}(x)u(s)\,ds-\int_{s_{0}}^{\infty}\mathfrak{g}^{s}(x)\mathfrak{u}(s)\,ds\right|\,dx\notag\\
=    &  \int_{-\infty}^{\infty}\left|\int_{s_{0}}^{\infty}g^{s}(x)u(s)\,ds-\int_{s_{0}}^{\infty}\mathfrak{g}^{s}(x)\mathfrak{u}(s)\,ds\pm\int_{s_{0}}^{\infty}g^{s}(x)\mathfrak{u}(s)\,ds\,\right|dx\notag\\
\leq &  \int_{-\infty}^{\infty}\left|\int_{s_{0}}^{\infty}g^{s}(x)u(s)\,ds-\int_{s_{0}}^{\infty}g^{s}(x)\mathfrak{u}(s)\,ds\right|\,dx\notag\\
     &  +\int_{-\infty}^{\infty}\left|\int_{s_{0}}^{\infty}g^{s}(x)\mathfrak{u}(s)\,ds-\int_{s_{0}}^{\infty}\mathfrak{g}^{s}(x)\mathfrak{u}(s)\,ds\right|\,dx\notag\\
=    &  \int_{-\infty}^{\infty}\left|\int_{s_{0}}^{\infty}g^{s}(x)\left\{u(s)-\mathfrak{u}(s)\right\}\,ds\right|\,dx\notag\\
     &  +\int_{-\infty}^{\infty}\left|\int_{s_{0}}^{\infty}\left\{g^{s}(x)-\mathfrak{g}^{s}(x)\right\}\mathfrak{u}(s)\,ds\right|\,dx\notag\\
\leq &  \int_{-\infty}^{\infty}\int_{s_{0}}^{\infty}g^{s}(x)\left|u(s)-\mathfrak{u}(s)\right|\,ds\,dx\notag\\
     & +\int_{-\infty}^{\infty}\left|\int_{s_{0}}^{\infty}\left\{\frac{g(x)}{G(s)}\mathbbm{1}_{(-\infty,s]}(x)-\frac{\mathfrak{g}(x)}{\mathfrak{G}(s)}\mathbbm{1}_{(-\infty,s]}(x)\right\}\mathfrak{u}(s)\,ds\right|\,dx\notag\\
=    &  \int_{s_{0}}^{\infty}\int_{-\infty}^{\infty}g^{s}(x)\,dx\left|u(s)-\mathfrak{u}(s)\right|\,ds\notag\\
     &  +\int_{-\infty}^{\infty}\left|\int_{s_{0}}^{\infty}\left\{\frac{g(x)}{G(s)}\mathbbm{1}_{(-\infty,s]}(x)-\frac{\mathfrak{g}(x)}{\mathfrak{G}(s)}\mathbbm{1}_{(-\infty,s]}(x)\right\}\mathfrak{u}(s)\,ds\right|\,dx\notag\\
=    &  \int_{s_{0}}^{\infty}\left|u(s)-\mathfrak{u}(s)\right|\,ds\notag\\
     &  +\int_{-\infty}^{\infty}\left|\int_{s_{0}}^{\infty}\left\{\frac{\mathfrak{G}(s)g(x)}{G(s)\mathfrak{G}(s)}\mathbbm{1}_{(-\infty,s]}(x)-\frac{G(s)\mathfrak{g}(x)}{G(s)\mathfrak{G}(s)}\mathbbm{1}_{(-\infty,s]}(x)\right\}\mathfrak{u}(s)\,ds\right|\,dx\notag\\
\leq &   \|u-\mathfrak{u}\|_{1}\notag\\
     &   +\int_{-\infty}^{\infty}\left|\int_{s_{0}}^{\infty}\left\{\mathfrak{G}(s)g(x)-G(s)\mathfrak{g}(x)\right\}\frac{\mathfrak{u}(s)}{G(s)\mathfrak{G}(s)}\mathbbm{1}_{(-\infty,s]}(x)\,ds\right|\,dx\notag\\
=    &  \|u-\mathfrak{u}\|_{1}\notag\\
     &  +\int_{-\infty}^{\infty}\left|\int_{s_{0}}^{\infty}\left\{\mathfrak{G}(s)g(x)-G(s)\mathfrak{g}(x)\pm\mathfrak{G}(s)\mathfrak{g}(x)\right\}\frac{\mathfrak{u}(s)}{G(s)\mathfrak{G}(s)}\mathbbm{1}_{(-\infty,s]}(x)\,ds\right|\,dx\notag\\
\leq &  \|u-\mathfrak{u}\|_{1}\notag\\
     &  +\int_{-\infty}^{\infty}\left|\int_{s_{0}}^{\infty}\left\{\mathfrak{G}(s)g(x)-\mathfrak{G}(s)\mathfrak{g}(x)\right\}\frac{\mathfrak{u}(s)}{G(s)\mathfrak{G}(s)}\mathbbm{1}_{(-\infty,s]}(x)\,ds\right|\,dx\notag\\
     &   +\int_{-\infty}^{\infty}\left|\int_{s_{0}}^{\infty}\left\{\mathfrak{G}(s)\mathfrak{g}(x)-G(s)\mathfrak{g}(x)\right\}\frac{\mathfrak{u}(s)}{G(s)\mathfrak{G}(s)}\mathbbm{1}_{(-\infty,s]}(x)\,ds\right|\,dx\notag\\
\leq &  \|u-\mathfrak{u}\|_{1}\notag\\
     &  +\int_{-\infty}^{\infty}\int_{s_{0}}^{\infty}\left|\mathfrak{G}(s)g(x)-\mathfrak{G}(s)\mathfrak{g}(x)\right|\frac{\mathfrak{u}(s)}{G(s)\mathfrak{G}(s)}\mathbbm{1}_{(-\infty,s]}(x)\,ds\,dx\notag\\
     &  +\int_{-\infty}^{\infty}\int_{s_{0}}^{\infty}\left|\mathfrak{G}(s)\mathfrak{g}(x)-G(s)\mathfrak{g}(x)\right|\frac{\mathfrak{u}(s)}{G(s)\mathfrak{G}(s)}\mathbbm{1}_{(-\infty,s]}(x)\,ds\,dx\notag\\
= &  \|u-\mathfrak{u}\|_{1}\notag\\
     &  +\int_{-\infty}^{\infty}\int_{s_{0}}^{\infty}\mathfrak{G}(s)\left|g(x)-\mathfrak{g}(x)\right|\frac{\mathfrak{u}(s)}{G(s)\mathfrak{G}(s)}\mathbbm{1}_{(-\infty,s]}(x)\,ds\,dx\notag\\
     &  +\int_{-\infty}^{\infty}\int_{s_{0}}^{\infty}\mathfrak{g}(x)\left|\mathfrak{G}(s)-G(s)\right|\frac{\mathfrak{u}(s)}{G(s)\mathfrak{G}(s)}\mathbbm{1}_{(-\infty,s]}(x)\,ds\,dx\notag\\
\leq &  \|u-\mathfrak{u}\|_{1}\notag\\
    &   +\int_{s_{0}}^{\infty}\int_{-\infty}^{\infty}\left|g(x)-\mathfrak{g}(x)\right|\mathbbm{1}_{(-\infty,s]}(x)\,dx\frac{\mathfrak{u}(s)}{G(s)\mathfrak{G}(s)}\,ds\notag\\
    &   +\int_{s_{0}}^{\infty}\int_{-\infty}^{\infty}\mathfrak{g}(x)\mathbbm{1}_{(-\infty,s]}(x)\,dx\left|\mathfrak{G}(s)-G(s)\right|\frac{\mathfrak{u}(s)}{G(s)\mathfrak{G}(s)}\,ds\notag\\
\leq &  \|u-\mathfrak{u}\|_{1}\notag\\
     &  +\|g-\mathfrak{g}\|_{1}\int_{s_{0}}^{\infty}\frac{\mathfrak{u}(s)}{G(s)\mathfrak{G}(s)}\,ds\notag\\
     &  +\sup_{s\in\mathbb{R}}\left|\mathfrak{G}(s)-G(s)\right|\int_{s_{0}}^{\infty}\frac{\mathfrak{u}(s)}{G(s)\mathfrak{G}(s)}\,ds\notag\\
\leq &  \|u-\mathfrak{u}\|_{1}+\left\{\|g-\mathfrak{g}\|_{1}+\sup_{s\in\mathbb{R}}\left|\mathfrak{G}(s)-G(s)\right|\right\}\int_{s_{0}}^{\infty}\frac{\mathfrak{u}(s)}{G(s)\mathfrak{G}(s)}\,ds\notag\\
\leq &  \|u-\mathfrak{u}\|_{1}+\frac{3}{2}\|g-\mathfrak{g}\|_{1}\int_{s_{0}}^{\infty}\frac{\mathfrak{u}(s)}{G(s)\mathfrak{G}(s)}\,ds. \label{eq:LemmaContIneq2b}
\end{align}
Putting together Equations \eqref{eq:LemmaContIneq1}--\eqref{eq:LemmaContIneq2b}, it follows that
\begin{align}
d_{1}\{(f,g),(\mathfrak{f},\mathfrak{g})\}  
\leq &   \|g-\mathfrak{g}\|_{1}\left\{2+\frac{3}{2}\int_{s_{0}}^{\infty}\frac{\mathfrak{u}(s)}{G(s)\mathfrak{G}(s)}\,ds\right\}+2\|u-\mathfrak{u}\|_{1}+2\left|\theta-\theta_{0}\right|+2\mathfrak{U}(s_{0}).\label{eq:LemmaContIneq3}
\end{align}
We are now in a position to find $\delta>0$ such that $d_2\{(g,u,\theta),(\mathfrak{g},\mathfrak{u},\theta_0)\}<\delta$ implies $d_1\{(f,g),(\mathfrak{f},\mathfrak{g})\}<\epsilon$. 
The inequality \eqref{eq:LemmaContIneq3} holds for any $s_0 \in \mathbbm{R}.$ Let $s_0$ be such that $\mathfrak{U}(s_{0})<\epsilon/8.$ If $\|g-\mathfrak{g}\|_{1} < \delta_1$  then $ \sup_{s\in\mathbb{R}}\left|\mathfrak{G}(s)-G(s)\right|<\delta_1/2,$ implying that $ G(s_{0})>\mathfrak{G}(s_{0})-\delta_1/2$ and $\mathfrak{G}(s_{0})G(s_{0})>\mathfrak{G}(s_{0})\left\{\mathfrak{G}(s_{0})-\delta_1/2\right\}$. Therefore, if $\|g-\mathfrak{g}\|_{1} < \delta_1 < 2 \mathfrak{G}(s_{0}),$ we have the following upper bound for the first term on the right hand side of \eqref{eq:LemmaContIneq3}: 
\begin{align*}
    \|g-\mathfrak{g}\|_{1}\left\{2+\frac{3}{2}\int_{s_{0}}^{\infty}\frac{\mathfrak{u}(s)}{G(s)\mathfrak{G}(s)}\,ds\right\} 
    & \leq \delta_1 \left[2+\frac{3}{2}\int_{s_{0}}^{\infty}\frac{\mathfrak{u}(s)}{\mathfrak{G}(s_{0})\left\{\mathfrak{G}(s_{0})-\delta_1/2\right\}}\,ds\right] \\ 
    & = \delta_1 \left[2+\frac{3}{2}\frac{1-\mathfrak{U}(s_0)}{\mathfrak{G}(s_{0})\left\{\mathfrak{G}(s_{0})-\delta_1/2\right\}}\right] \\
    & := \kappa(\delta_1).
\end{align*}
The function $\kappa$ is increasing in $\delta_1$ on the interval $(0,2\mathfrak{G}(s_{0}))$ with $\lim_{\delta_1 \downarrow 0} \kappa(\delta_1)=0,$ 
so for $\delta_1$ sufficiently small we have $\kappa(\delta_1) < \epsilon/4.$ 
If we set $\delta = \min\{\delta_1, \epsilon/8\},$ then
$d_2\{(g,u,\theta),(\mathfrak{g},\mathfrak{u},\theta_0)\}<\delta$ implies $d_1\{(f,g),(\mathfrak{f},\mathfrak{g})\}<\epsilon,$ as desired. 

\end{proof}


As in Subsection \ref{prior_description}, let $\text{DP}(P_0, \alpha)$ denote the Dirichlet process with base measure $P_0$ and precision parameter $\alpha >0,$ and let $\phi_{\mu, \sigma^2}$ denote the density of a Gaussian distribution with mean $\mu$ and variance $\sigma^2.$  Define the random density 
\begin{align*}
g(\cdot \mid P) = \int_{\mathbbm{R} \times \mathbbm{R}^+} \phi_{\mu,\sigma^2}(\cdot) P(d\mu,d\sigma^2), \quad P \sim \text{DP}(P_0, \alpha_1)
\end{align*} 
and let $\Pi_g$ be the probability measure associated with $g.$ 
\begin{lemma*}[Support] 
The probability measure $\Pi_g$ has full $L_1$-support on $\Delta$ provided that $P_0$ has full support on $\mathbb{R} \times \mathbb{R}^+.$ 
\end{lemma*}
\begin{proof}
Let $\mathfrak{g} \in \Delta$ and $\epsilon > 0$ be given. We need to demonstrate that $\Pi_g(\Vert \mathfrak{g} - g\Vert_1 < \epsilon)>0.$ Our main tools are Theorem 2 of \citet{Nguyen2020} and the stick-breaking representation of  Dirichlet process \citep{Sethuraman1994}. 

Theorem 2 of \citet{Nguyen2020} guarantees that there exists a finite mixture of Gaussian densities $\mathfrak{g}_m(x) = \sum_{j=1}^m \tilde c_j \phi_{\tilde \mu_j,\tilde \sigma_j^2}(x)$ 
such that $\Vert \mathfrak{g}_m - \mathfrak{g} \Vert_1<\epsilon/2$. To prove the lemma, it is sufficient to show that $\Pi_g(\Vert \mathfrak{g}_m - g \Vert_1 < \epsilon/2) > 0,$ because $\Pi_g(\Vert \mathfrak{g} - g\Vert_1 < \epsilon) \geq \Pi_g(\Vert \mathfrak{g}_m - g \Vert_1 < \epsilon/2).$ 

From the stick-breaking representation of the Dirichlet process, we have that
\begin{eqnarray*} 
g(x) = \sum_{j=1}^{\infty} c_j \phi_{\mu_j,\sigma_j^2}(x),
\end{eqnarray*}
where $c_j = v_{j}\prod_{l=1}^{j-1} (1-v_{l})$, each $v_j$ is an independent $\text{Beta}(1,\alpha)$ random variable, and each pair $(\mu_j,\sigma_j^2)$ is an independent random vector distributed according to the base measure $P_0$. 

Putting these observations together, it follows that
\begin{eqnarray*}
\Vert \mathfrak{g}_m - g \Vert_1 & \leq & \int_{-\infty}^\infty \left| \sum_{j=1}^m \tilde c_{j} \phi_{\tilde \mu_{j},\tilde \sigma_{j}^2}(x)  - \sum_{j=1}^m c_{j} \phi_{\mu_{j},\sigma_{j}^2}(x) \right| dx + \sum_{j=m+1}^{\infty} c_{j} \\
& = & \int_{-\infty}^\infty \left| \sum_{j=1}^m \tilde c_{j} \phi_{\tilde \mu_{j},\tilde \sigma_{j}^2}(x)  - \sum_{j=1}^m c_{j} \phi_{\mu_{j},\sigma_{j}^2}(x) \pm \sum_{j=1}^m \tilde c_{j} \phi_{\mu_{j},\sigma_{j}^2}(x) \right| dx + \sum_{j=m+1}^{\infty} c_{j} \\
& \leq & \sum_{j=1}^m  \left|\tilde c_{j} - c_{j} \right| +  \sum_{j=1}^m \tilde c_{j} \int_{-\infty}^\infty \left| \phi_{\tilde \mu_{j},\tilde \sigma_{j}^2}(x)  - \phi_{\mu_{j},\sigma_{j}^2}(x) \right| dx + \sum_{j=m+1}^{\infty} c_{j} \\
& \leq & \sum_{j=1}^m  \left|\tilde c_{j} - c_{j} \right| +  \sum_{j=1}^m \frac{|\tilde \mu_{j} - \mu_{j}|}{\tilde\sigma_{j}} +  \sum_{j=1}^m \frac{3|\tilde \sigma_{j}^2 - \sigma_{j}^2|}{\tilde\sigma_{j}^2}  + \sum_{j=m+1}^{\infty} c_{j}.
\end{eqnarray*} 
The last inequality follows from Theorem 1.3 of \citet{Devroye2020}.

We can use the inequality to find a condition on the first $m$ weights and atoms of the stick-breaking representation that guarantees $\|\mathfrak{g}_m - g \|_1 < \epsilon/2.$ In particular, if $\left|\tilde c_{j} - c_{j} \right|<\epsilon/8m$,  $|\tilde \mu_{j} - \mu_{j}| < \tilde\sigma_{j}\epsilon/8m$, and $|\tilde \sigma_{j}^2 - \sigma_{j}^2|<\tilde\sigma_{j}^2\epsilon/24m$ for each $j \in \{1,\ldots, m\}$, then 
\begin{align*}
    \sum_{j=1}^m  \left|\tilde c_{j} - c_{j} \right| <  \epsilon/8, \quad
    \sum_{j=1}^m \frac{|\tilde \mu_{j} - \mu_{j}|}{\tilde\sigma_{j}} < \epsilon/8, \quad  
    \sum_{j=1}^m \frac{3|\tilde \sigma_{j}^2 - \sigma_{j}^2|}{2\tilde\sigma_{j}^2} <  \epsilon/8 \end{align*}
and
\begin{align*}
    \sum_{j=m+1}^{\infty} c_{j}  = 1 - \sum_{j=1}^{m} c_{j} = \sum_{j=1}^{m} \tilde c_{j} - \sum_{j=1}^{m} c_{j} \leq \sum_{j=1}^{m} |\tilde c_{j} - c_{j}| < \epsilon/8,
\end{align*}
implying that $\Vert \mathfrak{g}_m - g \Vert_1<\epsilon/2$.

The map $S: (v_1,\ldots,v_m) \mapsto (c_1,\ldots,c_m)$ is continuous, which implies that $$S^{-1} \left[\{ (c_1,\ldots,c_m) : \left|\tilde c_{j} - c_{j} \right|<\epsilon/8m, \, j=1,\ldots,m\}\right]$$ is an open set. The prior distribution for $(v_1,\ldots,v_m)$ has full support and thus assigns positive mass to this set. 
The proof is completed by noticing that 
\begin{align*}
  & \hspace{-20mm} \Pi_g(\Vert \mathfrak{g}_m - g \Vert_1 < \epsilon/2 ) \geq \\  
  & \int_{S^{-1} \left[\{ (c_1,\ldots,c_m) : \left|\tilde c_{j} - c_{j} \right|<\epsilon/8m, \, j=1,\ldots,m\}\right]}
\prod_{j = 1}^m \frac{\Gamma(1+\alpha)}{\Gamma(\alpha)}(1-v_j)^{\alpha-1}\,  d v_1\ldots d v_m \\  
  &\times
\prod_{j = 1}^m P_0\left(|\tilde \mu_{j} - \mu_{j}| <\frac{\tilde\sigma_{j}\epsilon}{8m}, \,
|\tilde \sigma_{j}^2 - \sigma_{j}^2|<\frac{\tilde\sigma_{j}^2\epsilon}{24m}
\right) \, > \, 0.
\end{align*}

\end{proof}


Theorem~\ref{thm:support} follows easily from these lemmas. Suppose the assumptions A1-A3 of Theorem~\ref{thm:support} hold, let $(\mathfrak{f}, \mathfrak{g})\in LR_{\text{mix}}(\mathbbm{R})$ be given, and let $\mathbb{B}$ be a $d_1$-ball centered at $(\mathfrak{f}, \mathfrak{g}).$ We need to establish that $\Pi \circ T_1^{-1}$ assigns positive probability to $\mathbb{B}.$ By Theorem \ref{thm:fmixture_ac}, there exists $(\mathfrak{g} ,\mathfrak{u} , \theta_0) \in \Delta_{+} \times \Delta \times [0,1]$ such that $(\mathfrak{f}, \mathfrak{g}) = T_1(\mathfrak{g} ,\mathfrak{u} , \theta_0).$ The first lemma guarantees that $T_1^{-1}(\mathbb{B})$ contains a $d_2$-ball $\mathbb{H}$ centered at $(\mathfrak{g} ,\mathfrak{u} , \theta_0).$ It is sufficient to show that $\Pi$ assigns positive mass to $\mathbb{H}.$ Assumption A1 guarantees that $g,u,$ and $\theta$ are independent. That $\Pi(\mathbb{H}) > 0$ then follows from the full support of the prior for $\theta$ and the second lemma. 

\section*{Proof of Theorem \ref{thm:consistency}}

Suppose that the assumptions of Theorem \ref{thm:consistency} are satisfied and let $\epsilon > 0$ be given. If we can demonstrate that the Kullback-Leibler property \eqref{KLproperty} is satisfied, Theorem \ref{thm:consistency} follows from Schwartz's theorem \citep{Schwartz1965, Ghosal2017}. It is sufficient to show that there exists an open interval $\left(\underline{\theta},\overline{\theta}\right) \subset [0,1)$ and weak neighborhoods $\mathcal{V}_1$ and $\mathcal{V}_2$ on $\mathcal{P}(\mathbb{R} \times \mathbb{R}^+)$, the space of all probability measures defined on $\mathbb{R} \times \mathbb{R}^+$, such that $\left(\theta, P_1, P_2\right) \in \left(\underline{\theta},\overline{\theta}\right) \times \mathcal{V}_1 \times \mathcal{V}_2$ implies $\text{KL} \left( \mathfrak{m}, m \right) < \epsilon$. In that case, it follows from assumptions A1--A3  and Proposition 3 of \citet{Ferguson1973} that
\begin{equation}\label{KL_support_proof1}
\Pi \left\{  \text{KL} \left( \mathfrak{m}, m \right) < \epsilon \right\}
\geq
\Pi \left\{ \left(\underline{\theta},\overline{\theta}\right) \times \mathcal{V}_1 \times \mathcal{V}_2 \right\} 
> 0.    
\end{equation}
(Throughout the proof, we will use the notation $\Pi$ to refer to the prior measure rather than the more precise notation $\Pi \circ T_1^{-1}, \Pi \circ T_2^{-1},$ etc. to reduce the burden on the reader.)

One can show that
$$
\text{KL}(\mathfrak{m},m) 
 = q\cdot\text{KL}(\mathfrak{g},g) + (1-q)\cdot\text{KL}(\mathfrak{f},f)$$
which implies that
\begin{equation}\label{KL_support_proof2}
\Pi \left\{  \text{KL} \left( \mathfrak{m}, m \right) < \epsilon \right\}
>
\Pi \left\{ \text{KL} \left( \mathfrak{g}, g \right) < \frac{\epsilon}{2}, \,  \text{KL} \left( \mathfrak{f}, f \right) < \frac{\epsilon}{2}\right\}.
\end{equation}
By Theorem 3.2 in \citet{tokdar2006posterior} and assumptions B3--B6, for any $\epsilon_1>0$, there exists a weak neighborhood $\mathcal{V}_{1,1}$ such that $P_1 \in \mathcal{V}_{1,1}$ implies $\text{KL} \left( \mathfrak{g}, g \right) < \epsilon_1$. We will explicitly select $\epsilon_1$ at the end of this proof. For now, we focus on finding additional conditions on $(\theta, P_1, P_2)$ such that $\text{KL} \left( \mathfrak{f}, f \right)$ can be upper bounded by an arbitrarily small quantity less than $\epsilon/2$.  

Recall that
$$
f(x) = \theta g(x) + (1-\theta) \int_{-\infty}^{\infty} g^s(x) \, u(s) \,  ds
$$
and
$$
\mathfrak{f}(x) = \theta_0 \mathfrak{g}(x) + (1-\theta_0) \int_{-\infty}^{\infty} \mathfrak{g}^s(x) \, \mathfrak{u}(s) \,  ds.
$$
By Lemma 1 in \citet{Do2003}, we have that
\begin{align*}
     \text{KL} \left( \mathfrak{f}, f \right) & \leq
     \theta_0 \log\left(\frac{\theta_0}{\theta}\right) + (1-\theta_0) \log\left(\frac{1-\theta_0}{1-\theta}\right) + \\
     & \hspace{6mm} \theta_0 \text{KL} \left( \mathfrak{g}, g \right) +
     (1-\theta_0) \text{KL} \left\{ \int_{-\infty}^{\infty} \mathfrak{g}^s(\cdot) \, \mathfrak{u}(s) \,  ds, \int_{-\infty}^{\infty} g^s(\cdot) \, u(s) \,  ds \right\}\\
     & \leq
     \log\left(\frac{\theta_0}{\theta}\right) + \log\left(\frac{1-\theta_0}{1-\theta}\right) + \epsilon_1 + 
     \text{KL} \left\{ \int_{-\infty}^{\infty} \mathfrak{g}^s(\cdot) \, \mathfrak{u}(s) \,  ds, \int_{-\infty}^{\infty} g^s(\cdot) \, u(s) \,  ds \right\}
\end{align*}
when $P_1 \in \mathcal{V}_{1,1}.$ It follows from the inequality $\log(x) \leq x-1$ that 
\begin{align} \label{theta_ineq}
 & \log\left(\frac{\theta_{0}}{\theta}\right)+\log\left(\frac{1-\theta_{0}}{1-\theta}\right)  \leq \left(\frac{\theta_{0}}{\theta}-1\right)+\left(\frac{1-\theta_{0}}{1-\theta}-1\right).
\end{align}
For every $\epsilon_2>0$, we can construct an interval $\left(\underline{\theta},  \overline{\theta} \right)$ such that $\theta \in \left(\underline{\theta},  \overline{\theta} \right)$ implies 
\begin{align} \label{theta_inequality}
\left(\frac{\theta_{0}}{\theta}-1\right)+\left(\frac{1-\theta_{0}}{1-\theta}-1\right) \leq \epsilon_2 
\end{align}
as follows. For any $0< \delta_1 < \min(\theta_0,1-\theta_0)$ we have that $\left(\underline{\theta},  \overline{\theta} \right) = \left(\theta_0 - \delta_1, \theta_0 + \delta_1\right) \subset [0,1)$ and  
\begin{eqnarray}\nonumber
  \left(\frac{\theta_{0}}{\theta}-1\right)+\left(\frac{1-\theta_{0}}{1-\theta}-1\right) & < &\left(\frac{\theta_{0}}{\theta_0 - \delta_1}-1\right)+\left\{\frac{1-\theta_{0}}{1-(\theta_0 + \delta_1)}-1\right\} \\ \label{condition_theta}
  & = &\left\{\frac{\theta_{0}}{\theta_0 - \delta_1} + \frac{1-\theta_{0}}{1-(\theta_0 + \delta_1)}\right\}-2. 
\end{eqnarray}
As $\delta_1 \downarrow 0,$ the quantity \eqref{condition_theta} approaches zero. Thus, \eqref{theta_inequality} holds provided that $\delta_1$ is chosen to be sufficiently small. 
In that case, we have that
\begin{align} \label{KLf_ineq1}
     \text{KL} \left( \mathfrak{f}, f \right) &<
     \epsilon_1 + \epsilon_2 +
     \text{KL} \left\{ \int_{-\infty}^{\infty} \mathfrak{g}^s(\cdot) \, \mathfrak{u}(s) \,  ds, \int_{-\infty}^{\infty} g^s(\cdot) \, u(s) \,  ds \right\}
\end{align} when $\theta \in \left(\underline{\theta},  \overline{\theta} \right)$ and $P_1 \in \mathcal{V}_{1,1}.$

To further simplify the right hand side (RHS) of \eqref{KLf_ineq1}, we can apply the chain rule for the Kullback-Leibler divergence (see \citet[p.~92]{Wainwright2019} or \citet[p.~24]{Cover2006}). Let $\mathfrak{q}(x,s)$ and $q(x,s)$ be two joint densities on $\mathbbm{R}^2.$ The chain rule for the Kullback-Leibler divergence tells us that
\begin{eqnarray*}
\text{KL}\left\{\mathfrak{q}(x,s),q(x,s)\right\} & = & \text{KL}\left\{\mathfrak{q}(x), q(x) \right\}+\int_{-\infty}^{\infty}\text{KL}\left\{\mathfrak{q}(s \mid x),q(s \mid x)\right\} \mathfrak{q}(x) \, dx\\
 & = & \text{KL}\left\{\mathfrak{q}(s), q(s)\right\}+\int_{-\infty}^{\infty}\text{KL}\left\{\mathfrak{q}(x\mid s),q(x\mid s)\right\} \mathfrak{q}(s) \, ds
\end{eqnarray*}
which implies that
\begin{eqnarray*}
&  & \hspace{-15mm} \text{KL}\left\{\mathfrak{q}(x),q(x)\right\} \\
& = & \text{KL}\left\{\mathfrak{q}(x,s),q(x,s)\right\}-\int_{-\infty}^{\infty}\text{KL}\left\{\mathfrak{q}(s\mid x),q(s\mid x)\right\}\mathfrak{q}(x)\, dx\\
 & \leq & \text{KL}\left\{\mathfrak{q}(x,s),q(x,s)\right\}\\
 & = & \text{KL}\left\{\mathfrak{q}(s),q(s)\right\}+ \int_{-\infty}^{\infty}\text{KL}\left\{\mathfrak{q}(x\mid s),q(x\mid s)\right\} \mathfrak{q}(s) \, ds.
\end{eqnarray*}
Substituting in $\mathfrak{q}(x,s)=\mathfrak{g}^{s}(x)\mathfrak{u}(s)$ and $q(x,s)=g^{s}(x)u(s),$ we have 
\begin{align*}
\mathfrak{q}(x) &= \int_{-\infty}^{\infty}\mathfrak{g}^{s}(x)\mathfrak{u}(s) \, ds, \quad \mathfrak{q}(s) = \int_{-\infty}^{\infty}\mathfrak{g}^{s}(x)\mathfrak{u}(s) \, dx = \mathfrak{u}(s) \\
q(x) &= \int_{-\infty}^{\infty}g^{s}(x)u(s)\, ds, \quad q(s) = \int_{-\infty}^{\infty}g^{s}(x)u(s)\, ds = u(s) 
\end{align*}
with $\mathfrak{q}(x\mid s) = \mathfrak{g}^s(x)$ and $q(x\mid s) = g^s(x).$
Applying the previous inequality, we get that 
\begin{eqnarray*}
\text{KL} \left\{ \int_{-\infty}^{\infty} \mathfrak{g}^s(\cdot) \, \mathfrak{u}(s) \,  ds, \int_{-\infty}^{\infty} g^s(\cdot) \, u(s) \,  ds \right\} & \leq & \text{KL}\left(\mathfrak{u},u\right)+\int_{-\infty}^{\infty}\text{KL}\left\{\mathfrak{g}^{s}(\cdot),g^{s}(\cdot)\right\}\mathfrak{u}(s)\, ds.
\end{eqnarray*}
By Theorem 3.2 in \citet{tokdar2006posterior} and assumptions B3--B6, for every $\epsilon_3>0$, there exists a weak neighborhood $\mathcal{V}_2$ such that $P_2 \in \mathcal{V}_2$ implies $\text{KL} \left( \mathfrak{u}, u \right) < \epsilon_3$. Putting everything together, we see that 
\begin{align}
    \text{KL} \left( \mathfrak{f}, f \right) &<
     \epsilon_1 + \epsilon_2 +
     \text{KL} \left\{ \int_{-\infty}^{\infty} \mathfrak{g}^s(\cdot) \, \mathfrak{u}(s) \,  ds, \int_{-\infty}^{\infty} g^s(\cdot) \, u(s) \,  ds \right\} \nonumber \\
     & < \epsilon_1 + \epsilon_2 + 
     \text{KL}\left(\mathfrak{u},u\right)+\int_{-\infty}^{\infty}\text{KL}\left\{\mathfrak{g}^{s}(\cdot),g^{s}(\cdot)\right\}\mathfrak{u}(s)\, ds \nonumber \\
     & < \epsilon_1 + \epsilon_2 + \epsilon_3 +
     \int_{-\infty}^{\infty}\text{KL}\left\{\mathfrak{g}^{s}(\cdot),g^{s}(\cdot)\right\}\mathfrak{u}(s)\, ds \label{KLf_ineq2}
\end{align}
when $\theta \in \left(\underline{\theta}, \overline{\theta} \right), P_1 \in \mathcal{V}_{1,1},$ and $P_2 \in \mathcal{V}_2.$

The integral in \eqref{KLf_ineq2} can be expressed as 
\begin{eqnarray*}
 &  & \int_{-\infty}^{\infty}\text{KL}\left\{\mathfrak{g}^{s}(\cdot),g^{s}(\cdot)\right\}\mathfrak{u}(s)\, ds \\
 & = & \int_{-\infty}^{\infty}\int_{-\infty}^{\infty}\mathfrak{g}^{s}(x)\log\left\{\frac{\mathfrak{g}^{s}(x)}{g^{s}(x)}\right\} \, dx \, \mathfrak{u}(s)\, ds\\
 & = & \int_{-\infty}^{\infty}\int_{-\infty}^{s}\frac{\mathfrak{g}(x)}{\mathfrak{G}(s)}\log\left\{\frac{\frac{\mathfrak{g}(x)}{\mathfrak{G}(s)}}{\frac{g(x)}{G(s)}}\right\} \, dx \, \mathfrak{u}(s) \, ds\\
 & = & \int_{-\infty}^{\infty}\int_{-\infty}^{s}\frac{\mathfrak{g}(x)}{\mathfrak{G}(s)}\log\left\{\frac{\mathfrak{g}(x)}{g(x)}\frac{G(s)}{\mathfrak{G}(s)}\right\} \, dx \, \mathfrak{u}(s) \, ds\\
 & = & \int_{-\infty}^{\infty}\int_{-\infty}^{s}\frac{\mathfrak{g}(x)}{\mathfrak{G}(s)}\log\left\{\frac{\mathfrak{g}(x)}{g(x)}\right\} \, dx \, \mathfrak{u}(s) \, ds + \int_{-\infty}^{\infty}\int_{-\infty}^{s}\frac{\mathfrak{g}(x)}{\mathfrak{G}(s)} \, dx \, \log\left\{\frac{G(s)}{\mathfrak{G}(s)}\right\}\mathfrak{u}(s) \, ds\\
 & = & \underbrace{\int_{-\infty}^{\infty}\int_{-\infty}^{s}\frac{\mathfrak{g}(x)}{\mathfrak{G}(s)}\log\left\{\frac{\mathfrak{g}(x)}{g(x)}\right\}\, dx \, \mathfrak{u}(s) \, ds}_{\text{Term I}} + \underbrace{\int_{-\infty}^{\infty}\log\left\{\frac{G(s)}{\mathfrak{G}(s)}\right\} \mathfrak{u}(s) \, ds}_{\text{Term II}}.
\end{eqnarray*}
To complete the proof, we will demonstrate that both Term I and Term II can be bounded above by an arbitrarily small positive constant with positive prior probability.

First, we consider Term II. We have that
\begin{eqnarray} 
 \label{apply_log_ineq}
\int_{-\infty}^{\infty}\log\left\{\frac{G(s)}{\mathfrak{G}(s)}\right\}\mathfrak{u}(s)\, ds & \leq & \int_{-\infty}^{\infty}\left\{\frac{G(s)}{\mathfrak{G}(s)}-1\right\}\mathfrak{u}(s)\, ds \\
 \nonumber
 & = & \int_{-\infty}^{\infty}\frac{G(s)}{\mathfrak{G}(s)}\mathfrak{u}(s)\, ds - 1\\
 \nonumber
 & \leq & \int_{-\infty}^{\infty}\frac{\mathfrak{u}(s)}{\mathfrak{G}(s)}\, ds \\ 
 \label{apply_ratio}
 & = & \frac{\lim_{x \downarrow -\infty}\frac{\mathfrak{f}(x)}{\mathfrak{g}(x)} - \theta_0}{1-\theta_0} \\ 
 \label{apply_assumption_B2}
 & < & \infty
\end{eqnarray}
where \eqref{apply_log_ineq} follows from the inequality $\log(s) \leq s - 1,$ \eqref{apply_ratio} follows from equation \eqref{ratio} in the proof of Theorem \ref{thm:fmixture_ac}, and \eqref{apply_assumption_B2} follows from assumption B2. Because Term II is finite, for every $\epsilon_4>0$, there must exist $s_0 \in \mathbb{R}$ such that
\begin{eqnarray*}
\int_{-\infty}^{\infty}\log\left\{\frac{G(s)}{\mathfrak{G}(s)}\right\}\mathfrak{u}(s)\, ds & = & \int_{-\infty}^{s_{0}}\log\left\{\frac{G(s)}{\mathfrak{G}(s)}\right\}\mathfrak{u}(s) \, ds+\int_{s_0}^{\infty}\log\left\{\frac{G(s)}{\mathfrak{G}(s)}\right\}\mathfrak{u}(s) \, ds \\
 & \leq & \epsilon_4 +\int_{s_0}^{\infty}\log\left\{\frac{G(s)}{\mathfrak{G}(s)}\right\}\mathfrak{u}(s)\, ds.
\end{eqnarray*}
Applying the inequality $\log(s) \leq s-1$ once again, we get that 
\begin{eqnarray*}
\int_{s0}^{\infty}\log\left\{\frac{G(s)}{\mathfrak{G}(s)}\right\}\mathfrak{u}(s) \, ds & \leq & \int_{s0}^{\infty}\left\{\frac{G(s)}{\mathfrak{G}(s)}-1\right\}\mathfrak{u}(s) \, ds\\
 & \leq & \sup_{s\geq s_{0}}\left\{ \frac{G(s)}{\mathfrak{G}(s)}-1\right\}.
\end{eqnarray*}
We can find $\delta_{2}>0$ such that
$
\sup_{s\geq s_{0}} \left| G(s) - \mathfrak{G}(s) \right| < \delta_2
$
and   $P_1 \in \mathcal{V}_{1,1}$,
implies
\begin{align} \label{G_inequality}
\sup_{s\geq s_{0}}\left\{ \frac{G(s)}{\mathfrak{G}(s)}-1\right\} <  \frac{\sqrt{\epsilon_1}}{\mathfrak{G}(s_0)}
\end{align}
as follows. If $\left| G(s) - \mathfrak{G}(s) \right| < \delta_2$ for all $s  \geq s_0,$ then 
\begin{eqnarray*}
-\delta_{2}< & G(s)-\mathfrak{G}(s) & <\delta_{2}\\
\implies\quad\quad \mathfrak{G}(s)-\delta_{2}< & G(s) & <\mathfrak{G}(s)+\delta_{2}\\
\implies \frac{\mathfrak{G}(s)-\delta_{2}}{\mathfrak{G}(s)}-1< & \frac{G(s)}{\mathfrak{G}(s)}-1 & <\frac{\mathfrak{G}(s)+\delta_{2}}{\mathfrak{G}(s)}-1.
\end{eqnarray*}
Because
$$
\frac{\mathfrak{G}(s)+\delta_{2}}{\mathfrak{G}(s)}-1 = \frac{\delta_{2}}{\mathfrak{G}(s)}
$$
and 
$0 < \mathfrak{G}(s_0) \leq \mathfrak{G}(s)$ for all $s \geq s_0,$ we have that
$$
\frac{G(s)}{\mathfrak{G}(s)} -1
<\frac{\delta_{2}}{\mathfrak{G}(s)}
< \frac{\delta_{2}}{\mathfrak{G}(s_0)}, \quad s \geq s_0.
$$
Thus, the inequality \eqref{G_inequality} holds provided that $\delta_{2} < \sqrt{\epsilon_1}$.
The density $\mathfrak{g}$ is in the Kullback-Leibler support of the prior distribution associated with $g.$ It follows that $\mathfrak{g}$ is also in the $L_1$-support of the prior distribution associated with $g.$ The $L_1$-distance and total variation distance are equivalent metrics. Setting $\delta_{2} = \sqrt{\epsilon_1/2},$ Pinsker's inequality and the event
$P_1 \in \mathcal{V}_{1,1}$ implies  $$\sup_{s\geq s_{0}}\left|G(s)-\mathfrak{G}(s)\right| < \sqrt{\frac{\epsilon_1}{2}}$$ and thus
\begin{eqnarray*}
\int_{-\infty}^{\infty}\log\left(\frac{G(s)}{\mathfrak{G}(s)}\right)\mathfrak{u}(s)ds  & < & \epsilon_4 + \frac{\sqrt{\epsilon_1}}{\mathfrak{G}(s_0)}.
\end{eqnarray*} 

Next, we show that Term I is smaller than arbitrary $\epsilon_5>0$ with positive prior probability. The argument is somewhat involved and draws upon ideas from \citet{tokdar2006posterior} (the proofs of Theorem 3.2 and Lemma 3.1) and \citet{Ghosal1999} (the proof of Theorem 3). We introduce a sequence of densities $\left\{\tilde{\mathfrak{g}}(\cdot \, ;\, n) \mid n \geq 1\right\}$ defined on $\mathbb{R}$ and rewrite Term I as
\begin{align}
\nonumber
&\int_{-\infty}^{\infty}\int_{-\infty}^{s}\frac{\mathfrak{g}(x)}{\mathfrak{G}(s)}\log\left\{\frac{\mathfrak{g}(x)}{g(x)}\right\} dx \, \mathfrak{u}(s)\, ds \\ 
\nonumber
\quad=& \int_{-\infty}^{\infty}\int_{-\infty}^{\infty}\mathbbm{1}_{(-\infty,s]}(x)\mathfrak{g}(x)\log\left\{\frac{\mathfrak{g}(x)}{g(x)}\right\} dx \, \frac{\mathfrak{u}(s)}{\mathfrak{G}(s)} \, ds\\
\nonumber
 \quad=&  
 \nonumber \int_{-\infty}^{\infty}\int_{-\infty}^{\infty}\mathbbm{1}_{(-\infty,s]}(x)\mathfrak{g}(x)\log\left\{\frac{\mathfrak{g}(x)}{g(x)}\frac{\tilde{\mathfrak{g}}(x;n)}{\tilde{\mathfrak{g}}(x;n)}\right\} dx \, \frac{\mathfrak{u}(s)}{\mathfrak{G}(s)} \, ds\\
 \label{TermIII}
 \quad=&  \int_{-\infty}^{\infty}\int_{-\infty}^{\infty}\mathbbm{1}_{(-\infty,s]}(x)\mathfrak{g}(x)\log\left\{\frac{\mathfrak{g}(x)}{\tilde{\mathfrak{g}}(x;n)}\right\} dx \, \frac{\mathfrak{u}(s)}{\mathfrak{G}(s)} \, ds \\
 \label{TermIV}
 & + \int_{-\infty}^{\infty}\int_{-\infty}^{\infty}\mathbbm{1}_{(-\infty,s]}(x)\mathfrak{g}(x)\log\left\{\frac{\tilde{\mathfrak{g}}(x;n)}{g(x)}\right\} dx \, \frac{\mathfrak{u}(s)}{\mathfrak{G}(s)} \, ds.
\end{align}
It suffices to show that we can upper bound both \eqref{TermIII} and \eqref{TermIV} by $\epsilon_5/2$ with positive prior probability. 

We first consider the term \eqref{TermIII}. Nothing about this term is random. Our task is to find a sequence $\left\{\tilde{\mathfrak{g}}(\cdot \, ;\, n) \mid n \geq 1\right\}$ and a natural number $n_0$ such that
$$
\int_{-\infty}^{\infty}\int_{-\infty}^{\infty}\mathbbm{1}_{(-\infty,s]}(x)\mathfrak{g}(x)\log\left\{\frac{\mathfrak{g}(x)}{\tilde{\mathfrak{g}}(x;n_0)}\right\} dx \, \frac{\mathfrak{u}(s)}{\mathfrak{G}(s)} \, ds < \frac{\epsilon_5}{2}.
$$
In the proof of Theorem 3.2 from \citet{tokdar2006posterior}, the author shows that there exists a sequence of probability measures $\{\tilde{\mathfrak{P}}_n \mid n \geq 1\}$ with each $\tilde{\mathfrak{P}}_n$ defined on $[-n,n] \times \{\sigma^2_n\}$ such that the sequence $$\left\{\tilde{\mathfrak{g}}(\cdot \, ; \, n) := \int_{[-n,n]} \phi_{\mu,\sigma^2}(\cdot) \, \tilde{\mathfrak{P}}_n(d\mu,d\sigma^2) \Biggm| n \geq 1\right\}$$ satisfies 
$$
\lim_{n \rightarrow \infty} \log\left\{\frac{\mathfrak{g}(x)}{\tilde{\mathfrak{g}}(x;n)}\right\}
= 0, \quad x\in \mathbbm{R}
$$
and 
$$
\left|\log\left\{\frac{\mathfrak{g}(x)}{\tilde{\mathfrak{g}}(x;n)}\right\}\right| \leq \mathfrak{L}(x), \quad x\in \mathbbm{R}
$$
where $\mathfrak{L}(x)$ is a $\mathfrak{g}$ integrable function. It follows that
$$
\lim_{n \rightarrow \infty} 
\mathbbm{1}_{(-\infty,s]}(x)\log\left\{\frac{\mathfrak{g}(x)}{\tilde{\mathfrak{g}}(x;n)}\right\}
= 0, \quad s, x \in \mathbbm{R} 
$$
and 
$$
\mathbbm{1}_{(-\infty,s]}(x)\left|\log\left\{\frac{\mathfrak{g}(x)}{\tilde{\mathfrak{g}}(x;n)}\right\}\right| \leq \mathfrak{L}(x), \quad s, x \in \mathbbm{R}
$$
as well. We can conclude that $\mathfrak{L}(x)$ is $(x,s) \mapsto \mathfrak{g}(x) \frac{\mathfrak{u}(s)}{\mathfrak{G}(s)}$ integrable because it is $\mathfrak{g}$ integrable, it does not depend on $s,$ and $\int_{-\infty}^{\infty}\frac{\mathfrak{u}(s)}{\mathfrak{G}(s)}\, ds < \infty.$ By the dominated convergence theorem, we have that
$$
\lim_{n \rightarrow \infty} 
\int_{-\infty}^{\infty}\int_{-\infty}^{\infty}\mathbbm{1}_{(-\infty,s]}(x)\mathfrak{g}(x)\log\left\{\frac{\mathfrak{g}(x)}{\tilde{\mathfrak{g}}(x;n)}\right\} dx \, \frac{\mathfrak{u}(s)}{\mathfrak{G}(s)}\, ds = 0.
$$
Thus, we can find a natural number $n_{0,1}$ such that, for all $n \geq n_{0,1},$ the term \eqref{TermIII} is as small as desired, say less than $\epsilon_5/2$.  We will see at the conclusion of the proof, when we discuss how to choose the weak neighborhood $\mathcal{V}_1$, that we have an additional requirement related to this sequence. In particular, we also need to find $n_{0,2}$ such that, for all $n \geq n_{0,2}$,
\begin{equation} \label{V1_ineq_n0}
\int_{-\infty}^{\infty}\mathfrak{g}(x)\log\left\{\frac{\mathfrak{g}(x)}{\tilde{\mathfrak{g}}(x;n)}\right\}dx < \frac{\epsilon_5}{2} \left\{\int_{-\infty}^{\infty}\frac{\mathfrak{u}(s)}{\mathfrak{G}(s)}ds\right\}^{-1}.    
\end{equation}
Such an $n_{0,2}$ is guaranteed to exist because 
$$
\lim_{n \rightarrow \infty} 
\log\left\{\frac{\mathfrak{g}(x)}{\tilde{\mathfrak{g}}(x;n)}\right\}
= 0, \quad x \in \mathbbm{R}
$$
and 
$$
\left|\log\left\{\frac{\mathfrak{g}(x)}{\tilde{\mathfrak{g}}(x;n)}\right\}\right| \leq \mathfrak{L}(x), \quad x \in \mathbbm{R}.
$$
By setting $n_0 = \max(n_{0,1}, \, n_{0,2})$, we ensure that the term \eqref{TermIII} is less than $\epsilon_5/2$ and the inequality \eqref{V1_ineq_n0} holds.

We now let $n=n_0$ and turn our attention to establishing that the term \eqref{TermIV} is bounded above by $\epsilon_5/2$ with positive prior probability. Notice that we can rewrite this term as 
\begin{eqnarray}\nonumber
& & \hspace{-10mm} \int_{-\infty}^{\infty}\int_{-\infty}^{\infty}\mathbbm{1}_{(-\infty,s]}(x)\mathfrak{g}(x)\log\left\{\frac{\tilde{\mathfrak{g}}(x;n_0)}{g(x)}\right\}dx\,\frac{\mathfrak{u}(s)}{\mathfrak{G}(s)}ds \\\label{TermV}
& = & \int_{-\infty}^{\infty}\int_{|x|>k}\mathbbm{1}_{(-\infty,s]}(x)\mathfrak{g}(x)\log\left\{\frac{\tilde{\mathfrak{g}}(x;n_0)}{g(x)}\right\}dx \, \frac{\mathfrak{u}(s)}{\mathfrak{G}(s)}ds\\\label{TermVI}
& + & \int_{-\infty}^{\infty}\int_{-k}^{k}\mathbbm{1}_{(-\infty,s]}(x)\mathfrak{g}(x)\log\left\{\frac{\tilde{\mathfrak{g}}(x;n_0)}{g(x)}\right\} dx \, \frac{\mathfrak{u}(s)}{\mathfrak{G}(s)}ds.
\end{eqnarray}
Thus, bounding the term \eqref{TermV} and the term \eqref{TermVI} above by $\epsilon_5/4$ with positive probability is sufficient. 

To bound the term \eqref{TermV}, we use an argument adapted from the proof of Lemma 3.1 in \citet{tokdar2006posterior}. Choose positive values $\underline{\sigma}$ and $\overline{\sigma}$ such that $\underline{\sigma} < \overline{\sigma}$ and $\sigma^2_{n_0} \in (\underline{\sigma},\, \overline{\sigma})$, and select $k>n_0+\overline{\sigma}$ such that 
$$\int_{|x|>k}\mathfrak{g}(x)\left\{\frac{(|x|+n_{0})^{2}}{2\underline{\sigma}^{2}}\right\}dx < \frac{\epsilon_5}{4} \left\{\int_{-\infty}^{\infty}\frac{\mathfrak{u}(s)}{\mathfrak{G}(s)}ds\right\}^{-1}.$$ If $P_1 \in \mathcal{V}_{1,2} = \left\{P \in \mathcal{P}(\mathbb{R} \times \mathbb{R}^+) \mid P([-n_0,n_0]\times[\underline{\sigma},\overline{\sigma}])>\underline{\sigma}/\overline{\sigma}\right\}$, then
\begin{eqnarray}
\nonumber
 &  & \int_{-\infty}^{\infty}\int_{|x|>k}\mathbbm{1}_{(-\infty,s]}(x)\mathfrak{g}(x)\log\left\{\frac{\tilde{\mathfrak{g}}(x;n_0)}{g(x)}\right\} dx \, \frac{\mathfrak{u}(s)}{\mathfrak{G}(s)}ds\\\nonumber
 & = & \int_{-\infty}^{\infty}\int_{|x|>k}\mathbbm{1}_{(-\infty,s]}(x)\mathfrak{g}(x)\log\left\{\frac{\int_{[\underline{\sigma},\overline{\sigma}]}\int_{[-n_{0},n_{0}]}\phi_{\mu, \sigma^{2}}(x)\tilde{\mathfrak{P}}_{n_0}(d\mu,d\sigma^2)}{\int_{\mathbbm{R}\times\mathbbm{R}^{+}}\phi_{\mu,\sigma^{2}}(\cdot)P_{1}(d\mu,d\sigma^{2})}\right\} dx \, \frac{\mathfrak{u}(s)}{\mathfrak{G}(s)}ds\\\nonumber
 & \leq & \int_{-\infty}^{\infty}\int_{|x|>k}\mathbbm{1}_{(-\infty,s]}(x)\mathfrak{g}(x)\log\left\{\frac{\phi_{n_{0},\overline{\sigma}^{2}}(|x|)}{\phi_{-n_{0},\underline{\sigma}^{2}}(|x|)P_{1}(\ensuremath{\mathcal{V}_{1,4}})}\right\} dx \, \frac{\mathfrak{u}(s)}{\mathfrak{G}(s)}ds\\\nonumber
 & \leq & \int_{-\infty}^{\infty}\int_{|x|>k}\mathbbm{1}_{(-\infty,s]}(x)\mathfrak{g}(x)\left\{\frac{(|x|+n_{0})^{2}}{2\underline{\sigma}^{2}}\right\} dx \, \frac{\mathfrak{u}(s)}{\mathfrak{G}(s)}ds\\\nonumber
 & \leq & \int_{|x|>k}\mathfrak{g}(x)\left\{\frac{(|x|+n_{0})^{2}}{2\underline{\sigma}^{2}}\right\}dx \, 
 \int_{-\infty}^{\infty}\frac{\mathfrak{u}(s)}{\mathfrak{G}(s)}ds \\\label{V1_ineq_1}
 & < &  \frac{\epsilon_5}{4}.
\end{eqnarray}

Now we consider bounding the term \eqref{TermVI}. By the argument of \citet{Ghosal1999} referenced in the second paragraph of the proof of Lemma 3.1 in \citet{tokdar2006posterior}, 
there exists a neighborhood $\mathcal{V}_{1,3}$ such that $P_1 \in \mathcal{V}_{1,3},$ $\delta<1/3$, and $x \in [-n_0,n_0]$ imply that
$$
\left|\frac{\tilde{\mathfrak{g}}(x;n_0)}{g(x)}-1\right| < \frac{3\delta}{1-3\delta},
$$
which then implies that 
$$
\mathbbm{1}_{(-\infty,s]}(x)\left|\frac{\tilde{\mathfrak{g}}(x;n_0)}{g(x)}-1\right| < \frac{3\delta}{1-3\delta}.
$$ 
Combining this result with the inequality $\log(x) \leq x-1 \leq |x-1|$, we get that
\begin{eqnarray}\nonumber
 &  & \int_{-\infty}^{\infty}\int_{-k}^{k}\mathbbm{1}_{(-\infty,s]}(x)\mathfrak{g}(x)\log\left\{\frac{\tilde{\mathfrak{g}}(x;n_0)}{g(x)}\right\}dx \, \frac{\mathfrak{u}(s)}{\mathfrak{G}(s)}ds\\\nonumber
 & \leq & \int_{-\infty}^{\infty}\int_{-k}^{k}\mathbbm{1}_{(-\infty,s]}(x)\mathfrak{g}(x)\left|\frac{\tilde{\mathfrak{g}}(x;n_0)}{g(x)}-1\right|dx\frac{\mathfrak{u}(s)}{\mathfrak{G}(s)}ds\\\nonumber
 & \leq & \int_{-\infty}^{\infty}\int_{-k}^{k}\mathfrak{g}(x)\left|\frac{\tilde{\mathfrak{g}}(x;n_0)}{g(x)}-1\right|dx\frac{\mathfrak{u}(s)}{\mathfrak{G}(s)}ds\\\label{V1_ineq_2}
 & < & \frac{3\delta}{1-3\delta}\int_{-\infty}^{\infty}\frac{\mathfrak{u}(s)}{\mathfrak{G}(s)}ds.
\end{eqnarray} 
Therefore, by choosing $\delta$ such that $$\frac{3\delta}{1-3\delta}<\frac{\epsilon_5}{4} \left\{\int_{-\infty}^{\infty}\frac{\mathfrak{u}(s)}{\mathfrak{G}(s)}ds\right\}^{-1},$$ the term \eqref{TermVI} is bounded above by $\epsilon_5/4$. 

The results presented thus far demonstrate that 
$$
\text{KL} \left( \mathfrak{g}, g \right) < \epsilon_1
\,\mbox{ and }  
\text{KL} \left( \mathfrak{f}, f \right) < \epsilon_1 +  \frac{\sqrt{\epsilon_1}}{\mathfrak{G}(s_0)} + \epsilon_2 + \epsilon_3 + \epsilon_4 + \epsilon_5.
$$
We now describe how to choose $\epsilon_1, \ldots, \epsilon_5$, such that 
$$
\text{KL} \left( \mathfrak{g}, g \right) < \frac{\epsilon}{2}
\,\mbox{ and }  
\text{KL} \left( \mathfrak{f}, f \right) < \frac{\epsilon}{2}
$$
as desired in \eqref{KL_support_proof1} and \eqref{KL_support_proof2}. 

We first set $\epsilon_4 = \epsilon/10$, determining the value of $\mathfrak{G}(s_0)$. Subsequently, we select $\epsilon_5$, which leads to an upper bound for $\epsilon_1$ due to the relationship between $\mathcal{V}_{1,1}$, $\mathcal{V}_{1,2}$, and $\mathcal{V}_{1,3}$ as  weak neighborhoods of $\mathfrak{g}.$ Notice that $\mathfrak{P}_{n_0} \in \mathcal{V}_{1,2} \cap \mathcal{V}_{1,3}$, implying that $\mathcal{V}_{1,2} \cap \mathcal{V}_{1,3}$ is not empty and corresponds to an open set under the weak topology. Also, notice that $P_1 \in \mathcal{V}_{1,2} \cap \mathcal{V}_{1,3}$ implies
\begin{eqnarray*}
 \text{KL} \left( \mathfrak{g}, g \right) 
 & = & \int_{-\infty}^{\infty}\mathfrak{g}(x)\log\left\{\frac{\mathfrak{g}(x)}{g(x)}\right\}dx \\
 & = & \int_{-\infty}^{\infty}\mathfrak{g}(x)\log\left\{\frac{\mathfrak{g}(x)}{g(x)}\frac{\tilde{\mathfrak{g}}(x;n_{0})}{\tilde{\mathfrak{g}}(x;n_{0})}\right\}dx\\
 & = & \int_{-\infty}^{\infty}\mathfrak{g}(x)\log\left\{\frac{\mathfrak{g}(x)}{\tilde{\mathfrak{g}}(x;n_{0})}\right\}dx+\int_{-\infty}^{\infty}\mathfrak{g}(x)\log\left\{\frac{\tilde{\mathfrak{g}}(x;n_{0})}{g(x)}\right\}dx\\
 & = & \int_{-\infty}^{\infty}\mathfrak{g}(x)\log\left\{\frac{\mathfrak{g}(x)}{\tilde{\mathfrak{g}}(x;n_{0})}\right\}dx+\\
 &  & \int_{|x|>k}\mathfrak{g}(x)\log\left\{\frac{\tilde{\mathfrak{g}}(x;n_{0})}{g(x)}\right\}dx+\int_{-k}^{k}\mathfrak{g}(x)\log\left\{\frac{\tilde{\mathfrak{g}}(x;n_{0})}{g(x)}\right\}dx\\
 & < & \underbrace{\frac{\epsilon_5}{2}\left\{\int_{-\infty}^{\infty}\frac{\mathfrak{u}(s)}{\mathfrak{G}(s)}ds\right\}^{-1}}_{\mbox{by \eqref{V1_ineq_n0}}} + \underbrace{\int_{|x|>k}\mathfrak{g}(x)\left\{\frac{(|x|+n_{0})^{2}}{2\underline{\sigma}^{2}}\right\}dx}_{\mbox{see proof of Lemma 3.1 in \citet{tokdar2006posterior}}} + \\
 &  & \underbrace{\int_{-\infty}^{\infty}\int_{-k}^{k}\mathfrak{g}(x)\left|\frac{\tilde{\mathfrak{g}}(x;n_0)}{g(x)}-1\right|dx}_{\mbox{see proof of Theorm 3 in \citet{Ghosal1999}}}\\
  & < & \frac{\epsilon_5}{2}\left\{\int_{-\infty}^{\infty}\frac{\mathfrak{u}(s)}{\mathfrak{G}(s)}ds\right\}^{-1}+
  \underbrace{\frac{\epsilon_5}{2}\left\{\int_{-\infty}^{\infty}\frac{\mathfrak{u}(s)}{\mathfrak{G}(s)}ds\right\}^{-1}}_{\mbox{by \eqref{V1_ineq_1} and \eqref{V1_ineq_2}}}\\
 & = & \epsilon_5\left\{\int_{-\infty}^{\infty}\frac{\mathfrak{u}(s)}{\mathfrak{G}(s)}ds\right\}^{-1}.
\end{eqnarray*}
Therefore, we propose to choose $\epsilon_5$ such that 
 $$
 \min\left[
 \epsilon_5, 
 \epsilon_5\left\{\int_{-\infty}^{\infty}\frac{\mathfrak{u}(s)}{\mathfrak{G}(s)}ds\right\}^{-1}
 +
 \frac{1}{\mathfrak{G}(s_0)}\sqrt{\epsilon_5\left\{\int_{-\infty}^{\infty}\frac{\mathfrak{u}(s)}{\mathfrak{G}(s)}ds\right\}^{-1}}
 \right] < \frac{\epsilon}{10}
 $$ and  $\mathcal{V}_{1,1} \subset \mathcal{V}_{1,2} \cap \mathcal{V}_{1,3}$ with $\mathfrak{P}_{n_0} \in \mathcal{V}_{1,1}$. We propose to make $\mathcal{V}_1$ equal to  $\mathcal{V}_{1,1}$. With this choice of $\mathcal{V}_{1}$, we guarantee that $\epsilon_1 + \sqrt{\epsilon_1}/\mathfrak{G}(s_0) < \epsilon/10$, implying that, for $P_1 \in \mathcal{V}_{1}$,
$$
\text{KL} \left( \mathfrak{g}, g \right) < \frac{\epsilon}{10}
\,\mbox{ and }  
\text{KL} \left( \mathfrak{f}, f \right) < \frac{3\epsilon}{10} + \epsilon_2 + \epsilon_3.
$$

Finally, we propose setting $\epsilon_2$ and $\epsilon_3$ equal to $\epsilon/10$. By doing so, we determine $\left(\underline{\theta},\overline{\theta}\right)$, $s_0$, and $\mathcal{V}_2$. We now have that $(\theta, P_1, P_2) \in \left(\underline{\theta},\overline{\theta}\right) \times \mathcal{V}_1 \times \mathcal{V}_2$ implies $\text{KL} \left( \mathfrak{m}, m \right) < \epsilon$. By assumptions A1--A3  and Proposition 3 of \citet{Ferguson1973}, we have
$\Pi \left\{  \text{KL} \left( \mathfrak{m}, m \right) < \epsilon \right\}
\geq
\Pi \left\{ \left(\underline{\theta},\overline{\theta}\right) \times \mathcal{V}_1 \times \mathcal{V}_2 \right\} 
= \Pi \left\{ \left(\underline{\theta},\overline{\theta}\right) \right\}\times \Pi \left( \mathcal{V}_1\right) \times \Pi \left( \mathcal{V}_2 \right) > 0,$ as desired in \eqref{KL_support_proof1}. 

\newpage

\section*{Mixture representations for count distributions} \label{count_results}

There are analogous mixture representations for count distributions with countably infinite support. Let $\mathbbm{N} = \{1, 2, ... \}$  and $\mathbbm{N}_0 = \mathbbm{N} \cup \{0\}.$ Suppose that $f$ and $g$ are probability mass functions with $\sum_{x_i \in \mathbbm{N}_0} f(x_i) = 1, \sum_{x_i \in \mathbbm{N}_0} g(x_i) = 1,$ and $g > 0$ on $\mathbbm{N}_0.$ 

\begin{theorem*}[$f$ as a mixture]    The ratio $f/g$ is non-increasing on $\mathbb{N}_0$ if and only if there exists $\theta \in [0,1]$ and mass function $u$ with $\sum_{x_j \in \mathbbm{N}_0} u(x_j) = 1$ 
such that 
\begin{align}  \label{fmixture_count}
    f(x_i) &= \theta g(x_i) + (1-\theta) \sum_{x_j \in \mathbbm{N}_0}  g^{x_j}(x_i) u(x_j), \quad x_i \in \mathbbm{N}_0.
\end{align} Suppose a mixture representation \eqref{fmixture_count} exists. Then $\theta = \lim_{x_i \rightarrow \infty} f(x_i)/g(x_i).$ If $\theta \in [0,1),$ $u$ is also uniquely determined with 
\begin{align*}
    u(x_j) = \frac{G(x_j)}{1-\theta}\left\{\frac{f(x_j)}{g(x_j)} - \frac{f(x_{j+1})}{g(x_{j+1})}\right\}, \quad x_j \in \mathbb{N}_0.
\end{align*}
\end{theorem*}

\begin{theorem*}[$g$ as a mixture] 
The ratio $g/f$ is non-decreasing on $\mathbb{N}_0$ if and only if there exists $\omega \in [0,1]$ and a mass function $v$ with $\sum_{x_j \in \mathbb{N}} v(x_j) = 1$ such that
\begin{align} \label{gmixture_count}
    g(x_i) &= \omega f(x_i) + (1-\omega) \sum_{ x_j \in \mathbb{N}} f_{x_j}(x_i) \, v(x_j), \quad x_{i} \in \mathbb{N}_0.
\end{align} Suppose a mixture representation \eqref{gmixture_count} exists. Then $\omega=g(0)/f(0).$ If $\omega \in [0,1),$ $v$ is also uniquely determined with 
\begin{align*}
    v(x_j) = \frac{1-F(x_{j-1})}{1-\omega}\left\{\frac{g(x_j)}{f(x_j)} - \frac{g(x_{j-1})}{f(x_{j-1})} \right\}, \quad x_j \in \mathbb{N}.
\end{align*}
\end{theorem*}

A peculiarity of these results is that $u$ is supported on a subset of $\mathbbm{N}_0$ while $v$ is supported on a subset of $\mathbbm{N}.$ The proofs are omitted because they are nearly identical to the proofs of Theorems \ref{thm:fmixture_finite}  and \ref{thm:gmixture_finite}.   

\section*{Computational implementation} \label{computation_supplement}

As discussed in Subsection \ref{computation}, the posterior distribution is not available in closed form, but can be approximated via Markov chain Monte Carlo methods. We propose a slice-within-Gibbs sampler based on a truncated version of Sethuraman's stick-breaking representation of the Dirichlet process \citep{Sethuraman1994}. \texttt{R} code to implement the Markov chain Monte Carlo algorithm can be found at \href{https://github.com/michaeljauch/lrmix}{https://github.com/michaeljauch/lrmix}.

Under the truncated stick-breaking representation, the random measures $P_1$ and $P_2$ satisfy
$$
P_k(\cdot) = \sum_{j=1}^N v_{k,j} \left\{ \prod_{l =1}^{j-1} (1-v_{k,l}) \right\} \delta_{\mu_{k,j},\sigma_{k,j}^2}(\cdot), \quad k \in \{1,2\} 
$$
where the truncation level $N$ is fixed, $v_{k,j}$ are independent beta random variables for each $j \in \{1,\ldots,N-1\}$ with $v_{k, N} = 1,$ and the atoms $(\mu_{k,j},\sigma_{k,j}^2)^{\T}$ are independent random vectors distributed according to the base measure $P_{0,k}$. Let $\vec{v}_k = (v_{k,1},\ldots,v_{k,N-1})^{\T}$, $\vec{\mu}_k = (\mu_{k,1},\ldots,\mu_{k,N})^{\T}$,  and $
\vec{\sigma}_k^2 = (\sigma_{k,1}^2,\ldots,\sigma_{k,N}^2)^{\T}$. The densities of $G$ and $U$ are then
\begin{align*}
g(x \mid \vec{v}_1, \vec{\mu}_1, \vec{\sigma}_1^2) & = 
\sum_{j=1}^N v_{1,j} \left\{\prod_{l=1}^{j-1} (1-v_{1,l}) \right\} \phi_{\mu_{1,j},\sigma_{1,j}^2}(x), \\
u(x \mid \vec{v}_2, \vec{\mu}_2, \vec{\sigma}_2^2) & = 
\sum_{j=1}^N v_{2,j} \left\{\prod_{l=1}^{j-1} (1-v_{2,l})\right\} \phi_{\mu_{2,j},\sigma_{2,j}^2}(x),
\end{align*}
where $\phi_{\mu, \sigma^2}$ is the probability density function of a normal distribution with mean $\mu$ and variance $\sigma^2$. 

The spike-and-slab prior for $\theta$ discussed in Subsection \ref{prior_description} enables us to test $H_0: F = G$ against  $H_1: F \leq_{\text{LR}}G$. 
We set $\theta = (1-\gamma) \tilde \theta + \gamma$, where $\gamma$ is a binary random variable and $\tilde \theta$ is a random variable taking values in $(0,1)$. The hypotheses $H_0: F = G$ and $H_1: F \leq_{\text{LR}} G$ can be identified with the events $\gamma=1$ and $\gamma=0$, respectively. If a user is only interested in density estimation, they can use computational scheme outlined here but assign zero prior probability to $H_0$.

The density of $F$ can be expressed as
\begin{align*}
f(x \mid \gamma, \tilde \theta, \vec{v}_1, \vec{\mu}_1, \vec{\sigma}_1^2, \vec{v}_2, \vec{\mu}_2, \vec{\sigma}_2^2) & = \theta g(x \mid \vec{v}_1, \vec{\mu}_1, \vec{\sigma}_1^2) \, + \\ 
& \hspace{5mm} (1-\theta) \int_{-\infty}^{\infty} g^s(x \mid \vec{v}_1, \vec{\mu}_1, \vec{\sigma}_1^2) \, u(s \mid \vec{v}_2, \vec{\mu}_2, \vec{\sigma}_2^2) \, ds, 
\end{align*}
and the joint distribution of the data and parameters is given by
\begin{align}
    X_i \mid \gamma, \tilde \theta, \vec{v}_1, \vec{\mu}_1, \vec{\sigma}_1^2, \vec{v}_2, \vec{\mu}_2, \vec{\sigma}_2^2 & \stackrel{\mathrm{ind.}}{\sim} f(\cdot \mid \gamma, \tilde\theta, \vec{v}_1, \vec{\mu}_1, \vec{\sigma}_1^2, \vec{v}_2, \vec{\mu}_2, \vec{\sigma}_2^2), \quad i \in \{1,\ldots,n\}  \notag \\
    Y_i \mid \vec{v}_1, \vec{\mu}_1, \vec{\sigma}_1^2 & \stackrel{\mathrm{ind.}}{\sim}
    g(\cdot \mid \vec{v}_1, \vec{\mu}_1, \vec{\sigma}_1^2), \quad i \in \{1,\ldots,m\} \notag \\
    v_{1,j} & \stackrel{\mathrm{ind.}}{\sim} \mbox{Beta}(1,\alpha), \quad j\in\{1,\ldots,N-1\} \notag  \\
    (\mu_{1,j},\sigma_{1,j}^2)^{\T} & \stackrel{\mathrm{ind.}}{\sim} \mbox{Normal Inv-Gamma}(m,c,a_1,a_2), \quad j \in \{1,\ldots,N\} \notag  \\
    \tilde \theta \mid \gamma = 0 & \sim \mbox{Beta}(b_1,b_2)  \notag \\
    v_{2,j} \mid \gamma = 0 & \stackrel{\mathrm{ind.}}{\sim} \mbox{Beta}(1,\alpha), \quad j \in \{1,\ldots,N-1\}  \notag  \\
    (\mu_{2,j},\sigma_{2,j}^2)^{\T} \mid \gamma = 0 & \stackrel{\mathrm{ind.}}{\sim} \mbox{Normal Inv-Gamma}(m,c,a_1,a_2), \quad j \in \{1,\ldots,N\} \notag \\ 
    \tilde \theta \mid \gamma = 1 & \sim \mbox{Beta}(\breve b_1, \breve b_2)  \notag \\
    v_{2,j} \mid \gamma = 1 & \stackrel{\mathrm{ind.}}{\sim} \mbox{Beta}(1,\breve \alpha), \quad j \in \{1,\ldots,N-1\} \notag \\
    (\mu_{2,j},\sigma_{2,j}^2)^{\T} \mid \gamma = 1 & \stackrel{\mathrm{ind.}}{\sim} \breve p_{2,0}( \cdot), \quad j \in \{1,\ldots,N\}  \notag \\
    \gamma & \sim \mbox{Bernoulli}(p_0).  \label{HierarchicalModel} 
\end{align}
The density of the $\mbox{Normal Inv-Gamma}(m,c,a_1,a_2)$ distribution is
$$
\pi_{\mathrm{NI}}(\mu,\sigma^2 \mid m,c,a_1,a_2) =  \frac {\sqrt{c}} {\sigma\sqrt{2\pi} } \, \frac{a_2^{a_1}}{\Gamma(a_1)} \, \left( \frac{1}{\sigma^2} \right)^{a_1 + 1}   \exp \left\{ -\frac { 2 a_2 + c(\mu - m)^2} {2\sigma^2}  \right\},
$$
while the density of the $\mbox{Beta}(b_1,b_2)$ distribution is
$$
\pi_{\mathrm{Beta}}(v \mid b_1, b_2) =  \frac{\Gamma(b_1+b_2)}{\Gamma(b_1)\Gamma(b_2)} v^{b_1-1}(1-v)^{b_2-1}.
$$

The priors for $\tilde{\theta}, v_{2,j}$, and $(\mu_{2j}, \sigma^2_{2,j})^{\T}$ are defined conditionally on $\gamma$ for computational purposes. The priors given $\gamma = 1$ are called pseudo-priors \citep{Carlin1995}, and they help to ensure adequate mixing of the Markov chain. \citet{Carlin1995} proposed setting the hyperparameters of their pseudo-priors based on the data, which does not affect the stationary distribution of the Markov chain for the parameters of interest. The reader is referred to \citet{Carlin1995} for further details. A more complete discussion of our pseudo-priors appears later in this section. 

Below, we describe the full conditional distributions for the slice-within-Gibbs sampler. We use a dash as shorthand notation to avoid cumbersome enumeration of data and parameters, a common practice in the Bayesian literature.
\begin{enumerate}
    
    \item Update $(\vec{\mu}_1, \vec{\sigma}_1^2,\vec{v}_1)^{\T}$. For this step, we use the slice sampler described in Figure 8 of \citet{Neal2003} and the conditional densities
    \begin{align*}
        \pi(\vec{\mu}_1, \vec{\sigma}_1^2,\vec{v}_1 \mid \gamma = 0, -) 
        \propto & 
        \left\{\prod_{i=1}^{n} f(X_i \mid \gamma = 0, \tilde \theta, \vec{v}_1, \vec{\mu}_1, \vec{\sigma}_1^2, \vec{v}_2, \vec{\mu}_2, \vec{\sigma}_2^2) \right\}
        \\ & \times 
        \left\{\prod_{i=1}^{m} g(Y_i \mid \vec{v}_1, \vec{\mu}_1, \vec{\sigma}_1^2)\right\}
        \\ & \times 
        \left\{\prod_{j=1}^N \pi_{\mathrm{NI}}(\mu_{1,j},\sigma_{1,j}^2 \mid m,c,a_1,a_2)\right\}
        \left\{\prod_{j=1}^{N-1} \pi_{\mathrm{Beta}}(v_{1,j}|b_1,b_2)\right\} \\ 
        \pi(\vec{\mu}_1, \vec{\sigma}_1^2,\vec{v}_1 \mid \gamma = 1, -) &\propto
        \left\{\prod_{i=1}^{n} g(X_i \mid \vec{v}_1, \vec{\mu}_1, \vec{\sigma}_1^2)\right\}
        \left\{\prod_{i=1}^{m} g(Y_i \mid \vec{v}_1, \vec{\mu}_1, \vec{\sigma}_1^2)\right\}
        \\ & \times 
        \left\{\prod_{j=1}^N \pi_{\mathrm{NI}}(\mu_{1,j},\sigma_{1,j}^2 \mid m,c,a_1,a_2)\right\}
        \left\{\prod_{j=1}^{N-1} \pi_{\mathrm{Beta}}(v_{1,j}|b_1,b_2)\right\}.
    \end{align*}
    We evaluate $f(X_i \mid \gamma = 0, \tilde \theta, \vec{v}_1, \vec{\mu}_1, \vec{\sigma}_1^2, \vec{v}_2, \vec{\mu}_2, \vec{\sigma}_2^2)$ using numerical integration. 
    
    \item Update $(\vec{\mu}_2, \vec{\sigma}_2^2,\vec{v}_2)^{\T}$. 
    \begin{itemize}
        \item[i)] If $\gamma = 0$, we use the slice sampler described in Figure 8 of \citet{Neal2003} and the conditional density
        \begin{align*}
            \pi(\vec{\mu}_2, \vec{\sigma}_2^2,\vec{v}_2 \mid \gamma = 0, -) & \propto
            \left\{\prod_{i=1}^{n} f(X_i \mid \gamma = 0, \tilde \theta, \vec{v}_1, \vec{\mu}_1, \vec{\sigma}_1^2, \vec{v}_2, \vec{\mu}_2, \vec{\sigma}_2^2) \right\}
            \\ & \times 
            \left\{\prod_{j=1}^N \pi_{\mathrm{NI}}(\mu_{2,j},\sigma_{2,j}^2 \mid m,c,a_1,a_2)\right\}
            \left\{\prod_{j=1}^{N-1} \pi_{\mathrm{Beta}}(v_{2,j}| 1, \breve \alpha)\right\}.
        \end{align*}
        \item[ii)] If $\gamma = 1$, we sample
        \begin{align*}
            v_{2,j} \mid \gamma = 1, - & \sim \mbox{Beta}(1,\breve \alpha), \quad j \in \{1,\ldots,N-1\}
            \\
            (\mu_{2,j},\sigma_{2,j}^2) \mid \gamma = 1, - & \sim \breve p_{2,0}(\cdot), \quad j \in \{1,\ldots,N\}.
        \end{align*}
               
    \end{itemize}
    
    \item Update $\tilde \theta$. For each $i \in \{1,2, \, ... \, , n\},$  we introduce a latent variable $R_i$ that associates $X_i$ with one of the two components in the mixture representation \eqref{fx_mixture}. More precisely, 
    \begin{align*}
        X_i \mid R_i=1, \gamma \in \{0,1\}, - & \sim 
        g(\cdot \mid \vec{\mu}_1, \vec{\sigma}_1^2,\vec{v}_1) \\
        X_i \mid R_i=0, \gamma = 0, - & \sim 
        \int_{-\infty}^{\infty} g^s(\cdot \mid \vec{\mu}_1, \vec{\sigma}_1^2,\vec{v}_1) \, u(s \mid \vec{\mu}_2, \vec{\sigma}_2^2,\vec{v}_2) \, ds   \\
        R_i \mid \gamma = 0 &  \sim  \mbox{Bernoulli}(\tilde \theta) \\
        \tilde \theta \mid \gamma = 0 & \sim \mbox{Beta}(b_1,b_2) \\
        R_i \mid \gamma = 1 &  \sim  \delta_{1}(\cdot) \\
        \tilde \theta \mid \gamma = 1 & \sim \mbox{Beta}(\breve b_1, \breve b_2).
    \end{align*}
    Under our specification, $\mathrm{pr}(R_i = 0, \gamma = 1) = 0$. We update each $R_i$ with
    \begin{align*}
        \mathrm{pr}(R_i = 1 \mid \gamma = 0, -) & = \frac{\tilde \theta g(X_i \mid \vec{\mu}_1, \vec{\sigma}_1^2,\vec{v}_1)}{\tilde \theta g(X_i \mid \vec{\mu}_1, \vec{\sigma}_1^2,\vec{v}_1) + (1-\tilde \theta) \int_{-\infty}^{\infty} g^s(X_i \mid \vec{\mu}_1, \vec{\sigma}_1^2,\vec{v}_1)} \\
        \mathrm{pr}(R_i = 0 \mid \gamma = 0, -) &= 1 - \mathrm{pr}(R_i = 1 \mid \gamma = 0, -)\\ 
        \mathrm{pr}(R_i = 1 \mid \gamma = 1) &= 1.
    \end{align*}
    Then we update $\tilde \theta$ with
    \begin{align*}
        \tilde \theta \mid \gamma = 0, - &  \sim  \mbox{Beta}\left(b_1+\sum_{i=1}^{n} R_i,b_2+n-\sum_{i=1}^{n} R_i\right) \\
        \tilde \theta \mid \gamma = 1, - & \sim \mbox{Beta}(\breve b_1, \breve b_2).
    \end{align*}

    \item Update $\gamma$. The full conditional distribution for $\gamma$ is Bernoulli with 
    \begin{align*}
       \mathrm{pr}(\gamma = 0 \mid -) &= \frac{(1-p_0) \prod_{i=1}^{n} f(X_i \mid \gamma = 0, \tilde \theta, \vec{v}_1, \vec{\mu}_1, \vec{\sigma}_1^2, \vec{v}_2, \vec{\mu}_2, \vec{\sigma}_2^2)}{\mathcal{Z}_\gamma} \\
        \mathrm{pr}(\gamma = 1\mid  -) &= { \frac{p_0 \prod_{i=1}^{n} g(X_i|\vec{\mu}_1, \vec{\sigma}_1^2,\vec{v}_1)}{\mathcal{Z}_\gamma}},
    \end{align*}
where
    \begin{align*}
        \mathcal{Z}_\gamma &=  p_0 \prod_{i=1}^{n} g(X_i|\vec{\mu}_1, \vec{\sigma}_1^2,\vec{v}_1)+(1-p_0) \prod_{i=1}^{n} f(X_i \mid \gamma = 0, \tilde \theta, \vec{v}_1, \vec{\mu}_1, \vec{\sigma}_1^2, \vec{v}_2, \vec{\mu}_2, \vec{\sigma}_2^2).
    \end{align*}
\end{enumerate}

Follow the recommendations of \citet{Carlin1995}, we define the pseudo-priors $\mbox{Beta}(\breve b_1, \breve b_2)$, $\mbox{Beta}(1,\breve \alpha)$, and $\breve p_{2,0}(\cdot)$ to resemble the posterior distribution of $\tilde \theta$, $\vec{v}_2$, $\vec{\mu}_2$, and $\vec{\sigma}_k^2$ conditional on $\gamma = 0$. To achieve this, we take a two-stage approach. In the first stage, we run the first three steps of the slice-within-Gibbs sampler for $Q$ iterations with $\gamma$ fixed at zero. The Markov chain output $({\tilde \theta}_{(1)}, {\vec{v}}_{2,(1)}, {\vec{\mu}}_{2, \,(1)}, {\vec{\sigma}}^2_{2,(1)})^{\T}, \ldots, ({\tilde \theta}_{(Q)}, {\vec{v}}_{2,(Q)}, {\vec{\mu}}_{2, \,(Q)}, {\vec{\sigma}}^2_{2,(Q)})^{\T}$ provides an approximation of the posterior distribution of  $\tilde \theta$, $\vec{v}_2$, $\vec{\mu}_2$, and $\vec{\sigma}_2^2$  conditional on $\gamma = 0$. 
In the second stage, we use this Markov chain output to determine the pseudo-priors. For the pseudo-prior of $\tilde \theta$, we set
$$
(\breve b_1, \breve b_2) = \underset{(a,b)}{\rm argmax} \prod_{q=1}^Q \pi_{\rm Beta}\left({\tilde \theta}_{(q)}\left|a,b\right.\right).
$$
For the pseudo-prior of $\vec{v}_2$, we set 
$$
\breve \alpha = \frac{1}{Q}\sum_{q=1}^Q \underset{a}{\rm argmax} \prod_{j=1}^{N-1}  \pi_{\rm Beta}\left({v}_{2,j \, (q)}\left|1,a\right.\right)
$$
Finally, for the pseudo-prior of $(\vec \mu_{2}, \vec \sigma^2_{2})^{\T}$, we take $\breve p_{2,0}(\cdot)$ to be a kernel density estimate computed from $(m_1,s_1^2)^{\T},\ldots,(m_Q,s_Q^2)^{\T}$, where
$$
(m_q,s_q^2) \sim \sum_{j=1}^N  v_{2,j \, (q)} \left\{ \prod_{l = 1}^{j-1} (1- v_{2,l \, (q)}) \right\} \delta_{ \mu_{2,j \, (q)}, \sigma_{2,j \, (q)}^2}(\cdot).
$$
We compute the kernel density estimate using the function \texttt{kde} in the \texttt{R} package \texttt{ks} \citep{ksRpackage2021}. 

We conclude this section with recommendations for choosing the prior hyperparameters and the truncation level of the stick-breaking representation. Instead of working on the original scale of the data, we recommend standardizing. More precisely, we standardize by subtracting off the pooled sample mean and dividing by the pooled standard deviation. 
The density estimates can be transformed back to the original scale after running the slice-within-Gibbs sampler. We recommend $(b_1 = 1, b_2 = 1)$ and  $p_0 = 0.5$ as hyperparameters of the hierarchical prior for $\theta.$ 
We recommend $(m = 0, c = 1, a_1 = 1, a_2 = 1)$ as hyperparameters for the normal inverse-gamma base measure of the Dirichlet processes. For the precision parameter and the truncation level, we recommend $\alpha = 1$ and $N = 25.$ 
To verify that these choices are suitable in a given application, users should sample the densities $f$ and $g$ from the prior distribution, inspect them, and adjust as needed. 

\section*{Simulation study: Additional results}

In this section, we revisit the simulation study in Section 
3.6 to evaluate the performance of posterior means as point estimates. We consider the same scenarios, sample sizes, and number of simulations we considered in Section 3.6. For each scenario, sample size, and synthetic data set, we compute the following measures of performance:
      \begin{align*}
    L_2(\mathfrak{f}) &= \|E(f \mid X_1,\ldots, X_n, Y_1, \ldots, Y_m) - \mathfrak{f}\|_2\\
    L_2(\mathfrak{g}) &= \|E(g \mid X_1,\ldots, X_n, Y_1, \ldots, Y_m) - \mathfrak{g}\|_2\\
    L_2(\mathfrak{f}, \mathfrak{g}) &=  \lVert E(f \mid X_1,\ldots, X_n, Y_1, \ldots, Y_m) - \mathfrak{f} \rVert_2 + \lVert E(g \mid X_1,\ldots, X_n, Y_1, \ldots, Y_m) - \mathfrak{g} \rVert_2 \\
    L_2(\mathfrak{f}/\mathfrak{g}) &= \|E(f/g \mid X_1,\ldots, X_n, Y_1, \ldots, Y_m) - \mathfrak{f}/\mathfrak{g}\|_2.
    \end{align*}
 Since there is no guarantee that the $L_2$ distance between likelihood ratios is bounded, we compute $L_2(\mathfrak{f}/\mathfrak{g})$ only across the range of the observations:
        \begin{align*}
    L_2(\mathfrak{f}/\mathfrak{g}) &
    & = \int_{{\rm Range}(X_1,\ldots, X_n, Y_1, \ldots, Y_m)} \|E\left( \left. f(x)/g(x)\right| X_1,\ldots, X_n, Y_1, \ldots, Y_m\right) - \mathfrak{f}(x)/\mathfrak{g}(x) \|_2 \, dx .
    \end{align*}
    
The results are displayed in Table 2. The conclusion is that the monotonicity constraint is clearly advantageous because it consistently outperforms the unconstrained approach across all scenarios and sample sizes.

\begin{table}[h!]
\label{tab:simulationsl2mean}
\centering
\caption{Evaluation of posterior mean as an estimator through averaged $L_2$ losses}
\begin{tabular}{c|c|cccc|cccc|c}
         &   & \multicolumn{4}{c|}{LR order constrained model}  & \multicolumn{4}{c|}{Unconstrained model}  & Prob. of\\  
S & n & $\bar L_2(\mathfrak{f})$ & $\bar L_2(\mathfrak{g})$  & $\bar L_2(\mathfrak{f}, \mathfrak{g}) $  & $\bar L_2(\mathfrak{f}/\mathfrak{g}) $  & $\bar L_2(\mathfrak{f})$ & $\bar L_2(\mathfrak{g})$  & $\bar L_2(\mathfrak{f}, \mathfrak{g}) $  & $\bar L_2(\mathfrak{f}/\mathfrak{g})$ & $H_0:f=g$ \\ 
  \hline
I  &  100  &  0.059  &  0.076  &  0.135  &  2.449  &  0.059  &  0.087  &  0.146  &  2.942  &  0.002  \\
I  &  250  &  0.035  &  0.050  &  0.085  &  3.911  &  0.040  &  0.055  &  0.095  &  4.331  &  0.001  \\
I  &  500  &  0.027  &  0.037  &  0.064  &  4.644  &  0.030  &  0.039  &  0.070  &  5.049  &  0.001  \\ \hline
II  &  100  &  0.056  &  0.056  &  0.112  &  0.019  &  0.081  &  0.081  &  0.163  &  1.118  &  0.816  \\
II  &  250  &  0.035  &  0.035  &  0.071  &  0.013  &  0.049  &  0.053  &  0.102  &  0.924  &  0.850  \\
II  &  500  &  0.026  &  0.026  &  0.051  &  0.001  &  0.035  &  0.038  &  0.073  &  0.727  &  0.913  \\ \hline
III  &  100  &  0.137  &  0.147  &  0.284  &  2.909  &  0.188  &  0.218  &  0.406  &  3.987  &  0.062  \\
III  &  250  &  0.074  &  0.084  &  0.159  &  5.209  &  0.107  &  0.117  &  0.225  &  3.938  &  0.002  \\
III  &  500  &  0.053  &  0.061  &  0.114  &  4.395  &  0.072  &  0.074  &  0.146  &  5.100  &  0.001  \\ \hline
IV  &  100  &  0.115  &  0.115  &  0.231  &  0.046  &  0.175  &  0.180  &  0.355  &  1.854  &  0.821  \\
IV  &  250  &  0.070  &  0.070  &  0.139  &  0.027  &  0.098  &  0.103  &  0.201  &  1.151  &  0.855  \\
IV  &  500  &  0.052  &  0.052  &  0.104  &  0.013  &  0.067  &  0.067  &  0.134  &  1.191  &  0.896  \\ 
\end{tabular}
\end{table}

\section*{Additional application: Hoel (1972) data}

The C-reactive protein data were obtained through a data sharing agreement with the Children's Hospital of Philadelphia that does not permit the data to be made public. Therefore, we present an additional application based on the publicly available survival data analyzed in \citet{Hoel1972} and \citet{Carolan2005}. \texttt{R} code to reproduce this analysis can be found at \href{https://github.com/michaeljauch/lrmix}{https://github.com/michaeljauch/lrmix}.


The data come from an experiment in which 181 mice were exposed to radiation at the age of 5-6 weeks. The mice were divided into two groups, with $n=82$ living in a conventional lab environment and $m=99$ living in a germ-free environment. The outcome variable is time of death in days. We denote the distribution function of the outcome variable in the conventional environment by $F$ and the distribution function of the outcome variable in the germ-free environment by $G.$

Using the same methodology applied to the C-reactive protein data in Section \ref{subsec:CHOP}, we evaluate the posterior probability of the null hypothesis $H_0: F = G$ versus the alternative $H_1: F \le_{\text{LR}} G.$ Our point estimate of the probability that $F = G$ 
is $8 \times 10^{-4},$ which is consistent with the large test statistic and the rejection of $H_0$ at the $\alpha = 0.05$ significance level reported in \citet{Carolan2005}. These results should be interpreted with some caution, however, because the likelihood ratio order might not truly hold for these data, as discussed in \citet{Carolan2005}. Figure~\ref{fig:hoel_plot} is the analogue of Figure~\ref{fig:CHOP_plot} in the main text. 


\begin{figure}[htbp]
\centerline{\includegraphics[width=6in]{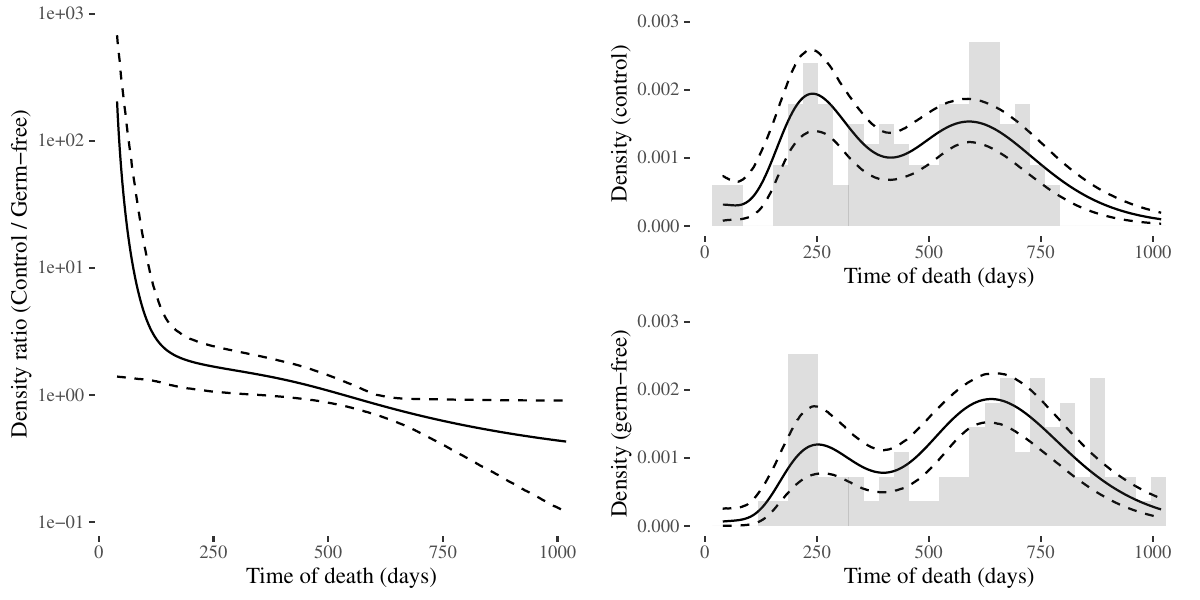}}
\caption{Left: Posterior mean of the density ratio $f/g$ (solid) along with 95\% pointwise credible intervals (dashed). Right: Density estimates (solid) and 95\% pointwise credible intervals (dashed) plotted on top of histograms of the outcome variable for the conventional environment (top) and the germ-free environment (bottom).}
\label{fig:hoel_plot}
\end{figure}



\end{document}